\documentclass[%
 reprint,
superscriptaddress,
nofootinbib,
 amsmath,amssymb,
 prx, 
 accepted=2019-09-27,
]{quantumarticle}

\usepackage{graphicx}
\usepackage{dcolumn}
\usepackage{bm}

\usepackage[numbers]{natbib}
\usepackage{graphicx}
\usepackage{subcaption}
\usepackage[colorlinks=true]{hyperref}
\usepackage{physics}
\usepackage[draft,
author=]{fixme}
\fxusetheme{color}
\makeatletter
\def\@lox@prtc{\section*{\@fxlistfixmename}\begingroup\def\@dotsep{4.5}}
\def\@lox@psttc{\endgroup}
\makeatother

\usepackage[english]{babel}

\usepackage{tikz}
\usetikzlibrary{shapes,arrows,positioning} 
\tikzstyle{decision} = [shape=rectangle, draw=black, very thick,
text width=1.5em, text centered,minimum height=2em]
\tikzstyle{chance} = [draw=black, very thick, shape=circle,
text width=1.1em, text badly centered]
\tikzset{
  >=stealth',
  pil/.style={
    ->,
    thick, 
    shorten <=2pt, 
    shorten >=2pt,}
}

\usepackage{qcircuit}

\usepackage{amsthm}
\newtheorem{definition}{Definition}
\newtheorem{theorem}{Theorem}
\newtheorem{lemma}{Lemma}

\newcommand{\Z}{\mathbb{Z}}

\newcommand{\M}{\mathcal{M}}
\newcommand{\Id}{\mathbb{I}}
\newcommand{\supp}{\textup{Supp}}
\newcommand{\m}{{\left| \mathcal M \right |}}
\newcommand{\olra}{\overleftrightarrow}
\newcommand{\ola}{\overleftarrow}
\newcommand{\ora}{\overrightarrow}

\renewcommand{\Pr}{\textup{Pr}}

\newcommand{\abclz}{{ABCL\textsubscript{0}} }
\newcommand{\abclo}{{ABCL\textsubscript{1}} }

\newcommand{\proj}[1]{{\bm[}#1{\bm]}}

\begin{document}

\title{$\psi$-epistemic interpretations of quantum theory have a measurement problem}

\author{Joshua B. Ruebeck}
\email{jbruebeck@uwaterloo.ca}
\orcid{0000-0003-3671-1667}
\author{Piers Lillystone}
\affiliation{%
  Institute for Quantum Computing and Department of Physics and Astronomy\\
  University of Waterloo, Waterloo, Ontario N2L 3G1, Canada
}%
\author{Joseph Emerson}
\affiliation{
  Institute for Quantum Computing and Department of Applied Math\\
  University of Waterloo, Waterloo, Ontario N2L 3G1, Canada
}
\affiliation{
  Canadian Institute for Advanced Research, Toronto, Ontario M5G 1Z8, Canada
}
\date{\today}

\begin{abstract}
  $\psi$-epistemic interpretations of quantum theory maintain that
  quantum states only represent incomplete information about the
  physical states of the world.  A major motivation for this view is
  the promise to provide a reasonable account of state update under
  measurement by asserting that it is simply a natural feature of
  updating incomplete statistical information.  Here we demonstrate
  that all known $\psi$-epistemic ontological models of quantum theory
  in dimension $d\geq3$, including those designed to evade the 
  conclusion of the PBR theorem, cannot represent state update
  correctly. Conversely, interpretations for which the wavefunction is
  real evade such restrictions despite remaining subject to
  long-standing criticism regarding physical discontinuity,
  indeterminism and the ambiguity of the Heisenberg cut. This revives
  the possibility of a no-go theorem with no additional assumptions,
  and demonstrates that what is usually thought of as a strength of
  epistemic interpretations may in fact be a weakness.
\end{abstract}

\maketitle


\section{Introduction}
\label{sec:introduction}

There are many interpretations of quantum theory\footnote{This is an
  understatement.}. Among the many differences between these
interpretations, one that often takes center stage is the stance that
they take towards the wavefunction or quantum state.  Three broad
categories have been identified which capture a number of
interpretations.  Two of these categories are more commonly
juxtaposed: $\psi$-ontic
interpretations~\cite{BohmSuggestedInterpretationQuantum1952,EverettRelativeStateFormulation1957,DiracPrinciplesQuantumMechanics1958,GhirardiUnifieddynamicsmicroscopic1986,PearleCombiningstochasticdynamical1989,Valentinipilotwavetheoryclassical1991,
  BohmUndividedUniverse1993,HollandQuantumTheoryMotion1993,Beltrametticlassicalextensionquantum1995,DieksModalInterpretationQuantum1998,PearleCollapsemodels1999,BarrettQuantumMechanicsMinds1999,LombardiModalInterpretationsQuantum,BassiModelswavefunctioncollapse2013}
posit that the quantum state is a part of the \emph{real} (physical)
state of a system, whereas $\psi$-epistemic
interpretations~\cite{EinsteinCanQuantumMechanicalDescription1935,KochenProblemHiddenVariables1967,BallentineStatisticalInterpretationQuantum1970,HowardEinsteinlocalityseparability1985,BallentineInadequacyEhrenfesttheorem1994,EmersonQuantumChaosQuantumClassical2002,RudolphOntologicalModelsQuantum2006,Spekkensdefenseepistemicview2007,HarriganEinsteinincompletenessepistemic2010}
argue that the quantum state is merely a state of \emph{knowledge}
about the real state of the system. A very thorough review of these
two stances can be found in~\cite{Leiferquantumstatereal2014}. A third
recently articulated category of $\psi$-doxastic
interpretations~\cite{CabelloInterpretationsquantumtheory2017,SchackParticipatoryrealism2016,WernerHeisenbergPhysicsPhilosophy1958,FuchsintroductionQBismapplication2014,RovelliRelationalquantummechanics1996,HealeyQuantumTheoryPragmatist2010,BubWhyBohrwas2017}
argue that the quantum state is a state of \emph{belief}, and are
distinguished from $\psi$-epistemic interpretations by the fact that
they deny that a system has some `real state.' While not all
interpretations conform to these three descriptors, they are useful
categories insofar as they allow us to qualitatively discuss certain
features separately from the particular interpretation in which they
are embedded.

The $\psi$-epistemic class has garnered attention as a view which
provides very appealing explanations of otherwise paradoxical features
of quantum theory like the state update
rule~\cite{BallentineStatisticalInterpretationQuantum1970}, the
classical limit under quantum chaos
\cite{EmersonQuantumChaosQuantumClassical2002}, no
cloning~\cite{FuchsQuantumMechanicsQuantum2002}, and
entanglement~\cite{GrossHudsonTheoremfinitedimensional2006}. For
example, through this viewpoint state update is not a physical
`collapse' process and therefore not subject to paradoxes,
indeterminism, and discontinuity; rather it is understood as analogous
to the non-pardoxical `collapse' of a subjective probability
distribution via Bayes' rule upon consideration of new information.
While several $\psi$-epistemic models have been
proposed~\cite{KochenProblemHiddenVariables1967,MontinaDynamicsqubitclassical2012,LewisDistinctQuantumStates2012,Aaronsonpsepistemictheoriesrole2013,HarriganRepresentingprobabilisticdata2007,Spekkensdefenseepistemicview2007,WallmanNonnegativesubtheoriesquasiprobability2012},
they generally have undesirable features or are restricted to a
subtheory of full quantum theory.  None achieve all of the features
that an optimistic $\psi$-epistemicist would expect.

This suggests the possibility that a fully satisfactory
$\psi$-epistemic interpretation cannot actually explain all of quantum
theory despite the qualitatively compelling features of such a
view\footnote{As is well-known, the work of
  Bell~\cite{BellProblemHiddenVariables1966} has shown that no
  $\psi$-epistemic interpretation can evade the non-locality that
  manifests trivially in $\psi$-ontic interpretations; this trivial
  manifestation of non-locality in $\psi$-ontic interpretations is an
  oft-forgotten insight from
  Einstein~\cite{HarriganEinsteinincompletenessepistemic2010,HowardEinsteinlocalityseparability1985,EinsteinCanQuantumMechanicalDescription1935}.}. This
suspicion has led to a number of no-go theorems in recent years which
establish that, given at least one additional assumption, any
consistent interpretation of quantum theory cannot be
$\psi$-epistemic~\cite{Puseyrealityquantumstate2012,Aaronsonpsepistemictheoriesrole2013,Colbecksystemwavefunction2017,HardyAREQUANTUMSTATES2013,MansfieldRealityquantumstate2016}.
These no-go theorems are generally proven within the ontological
models formalism, which describes a large class of existing
interpretations of quantum
theory~\cite{SpekkensContextualitypreparationstransformations2005,HarriganOntologicalmodelsinterpretation2007,RudolphOntologicalModelsQuantum2006}.

Within the ontological models formalism\footnote{A note on potentially
  confusing terminology---an ontological model assumes the reality of
  \emph{some} kind of state (hence ``ontological''), but does not
  assume the reality of the quantum state specifically (and so is not
  necessarily $\psi$-ontic).}, $\psi$-epistemic models can be given a
precise mathematical definition called the $\psi$-epistemic criterion
\cite{Spekkensdefenseepistemicview2007}. This criterion allows the
possibility of conclusively ruling out this type of model. Outside of
this framework, it is unlikely that $\psi$-epistemic models can be
precluded with any kind of certainty; $\psi$-doxastic interpretations,
for example, do not fit neatly into the ontological models framework
and thus are not necessarily ruled out by these no-go theorems. This
is despite the fact that they share many of the features which make
$\psi$-epistemic interpretations appealing. In the present paper we
restrict our attention to the ontological models framework.

The fact that an extra assumption is required to rule out
$\psi$-epistemic theories has purportedly been demonstrated by the
existence of $\psi$-epistemic models which, while being individually
unsatisfactory for various reasons, do satisfy at least the bare
minimum requirements of a $\psi$-epistemic
theory~\cite{LewisDistinctQuantumStates2012,Aaronsonpsepistemictheoriesrole2013,HarriganRepresentingprobabilisticdata2007}. All
of these models were specified within a prepare-measure framework, so
they have been proven to reproduce quantum statistics for all
experiments that involve preparing a state and then measuring it
once. In this paper we show that, if we allow sequential measurements
in the operational description, these models cannot reproduce
operational statistics.

Our main contribution in this work is thus to demonstrate that the
state update rule imposes severe constraints on $\psi$-epistemic
models. This is in contrast to the prevailing view that, as
articulated by Leifer, ``a straightforward resolution of the collapse
of the wavefunction, the measurement problem, Schr\"odinger’s cat and
friends is one of the main advantages of $\psi$-epistemic
interpretations''~\cite{Leiferquantumstatereal2014}. As a consequence,
we revive the possibility of a general no-go theorem for
$\psi$-epistemic models that doesn't rely on an additional assumption
such as the locality assumption required
in~\cite{Puseyrealityquantumstate2012} which conflicts with the
non-locality that is implied by Bell's
theorem~\cite{Emersonwholegreatersum2013,MansfieldRealityquantumstate2016}. See
Appendix~\ref{sec:additional-assumption} for further discussion on
this point.

Although state update under measurement has been described in a few
specific
models~\cite{MontinaDynamicsqubitclassical2012,Spekkensdefenseepistemicview2007,CataniSpekkenstoymodel2017,LillystoneContextualpsEpistemicModel2019}
and discussed with regards to
contextuality~\cite{KaranjaiContextualityboundsefficiency2018}, it has
yet to be treated generally or in relation to the
$\psi$-epistemic/$\psi$-ontic distinction\footnote{Although the
  Leggett-Garg
  inequalities~\cite{LeggettQuantummechanicsmacroscopic1985,EmaryLeggettGargInequalities2014}
  might be construed as a general treatment of state update, it is
  more accurate to say that they are about the \emph{absence} of state
  update.}. Here we take some preliminary steps in both of these
directions, and argue that $\psi$-epistemic models are the natural
arena in which to investigate interesting behavior of state update
under measurement.

In Section~\ref{sec:defin-meas-update}, we describe the ontological
models formalism, adding a description of state update under
measurement and motivating its importance. Despite this motivation,
one might still argue that many distinctly quantum phenomena
(e.g. Bell inequality violations) can be described without reference
to state update; thus, from an operationalist point of view, we
shouldn't need to consider state update in order to investigate these
phenomena.  However, we show in Section~\ref{sec:main-thm} that the
consideration of state update actually places nontrivial restrictions
on how one can represent even a prepare-and-measure-once
experiment. Thus our results are directly applicable to models which
have only specified behavior for a single
measurement. Section~\ref{sec:exampl-ontol-models} reviews a number of
examples of ontological models from the literature; in each case we
either specify its state update rule (in dimension $d=2$, for
$\psi$-ontic models, and for some models of subtheories) or prove its
impossibility (for all known $\psi$-epistemic models in dimension
$d\geq3$). Finally, we discuss the implications of our results and
describe some open questions in Section~\ref{sec:discussion}.

\section{Defining measurement update in ontological models}
\label{sec:defin-meas-update}

\subsection{The ontological models formalism}
\label{sec:ont-models}

In the standard treatment, an \emph{operational
  theory}~\cite{SpekkensContextualitypreparationstransformations2005}
is described by a set of preparations $\mathcal P$, a set of
transformations $\mathcal T$, and a set of measurements $\mathcal M$
along with a probability distribution
\begin{equation}
  \label{eq:40}
  \Pr(k|M,T,P).
\end{equation}
This quantity describes the probability of some measurement outcome
$k\in\Z$ given an experimenter's choice of $P\in\mathcal P$,
$T\in\mathcal T$, and $M\in\mathcal M$. When considering
transformations this is called the \emph{prepare-transform-measure}
operational framework, and when we omit transformations it is the
\emph{prepare-measure} framework. Often we take $\mathcal P$ to be the
set of pure quantum state preparations, $\mathcal T$ to be the full
set unitary maps on a Hilbert space, and $\mathcal M$ to be all
projective measurements on this Hilbert space. In this case, we say we
are describing the full quantum theory\footnote{Larger sets can be
  considered (e.g. including mixed states, CPTP maps, or
  non-projective measurements), but all of the models studied in this
  paper fit the given definition.}; in contrast, a \emph{subtheory} is
described by taking subsets of $\mathcal P,\mathcal T,\mathcal M$ for
the full quantum theory. For example, in quantum information settings
we often consider only measurements in the standard basis.

Note that this standard definition involves a single measurement and a
single measurement outcome despite the fact many important quantum
experiments (e.g. Stern-Gerlach, double
slit~\cite{SakuraiModernquantummechanics2011}) and quantum algorithms
(e.g. measurement-based error
correction~\cite{GottesmanIntroductionQuantumError2009}) involve
multiple measurements. Thus we will refer to the usual definition of
prepare-measure as prepare-and-measure-once experiments. In this
paper, we are concerned with multiple measurements, so we will also
have to describe probabilities like
\begin{equation}
  \label{eq:41}
  \Pr(k_2,k_1|M_2,M_1,P).
\end{equation}

Additionally, we note that positive operator valued measures (POVMs)
do not fully specify how a measurement updates a state. Although one
can obtain a POVM $\{E_k\}$ from a set of generalized measurement
operators $\{M_k\}$ by the relation $E_k=M_k^\dag M_k$, the
decomposition of $\{E_k\}$ into $\{M_k\}$ is not unique.  Thus
although $\mathcal M$ is often described by POVMs, consideration of
state update requires that we specify generalized measurement
operators instead. As an example of when this is important, consider a
coarse-graining of the measurement $\{M_k\}=\{\proj0,\proj1,\proj2\}$,
where we denote the projector onto a state
\begin{equation}
  \proj\psi=\ketbra{\psi}{\psi}\label{eq:14}
\end{equation}
as in~\cite{Leiferquantumstatereal2014}. We can either coarse-grain
coherently, i.e. measure $\{M_k'\}=\{\proj0+\proj1,\proj2\}$ or we can
coarse-grain decoherently by measuring $\{M_k\}$ and then combining
outcomes 0 and 1 into a single measurement result and `forgetting'
which one actually occurred. While these two processes are represented
by the same POVM $\{E_k\}=\{\proj0+\proj1,\proj2\}$, their state
update behavior is different: if the state
$\frac{1}{\sqrt{2}}\qty(\ket0+\ket1)$ is measured, it will stay the
same in the coherent case or update to the mixed state
$\frac{1}{2}\qty(\proj0+\proj1)$ in the decoherent case.

An \emph{ontological
  model}~\cite{SpekkensContextualitypreparationstransformations2005,HarriganOntologicalmodelsinterpretation2007,RudolphOntologicalModelsQuantum2006}
supplements this operational point of view by asserting that a system
has a state $\lambda$, called an \emph{ontic state}. To specify an
ontological model, we first choose an \emph{ontic state space}
$\Lambda$. Then, preparations are described by a \emph{preparation
  distribution} $\mu(\lambda|P)$, which is the probability of
preparing some state $\lambda\in\Lambda$ given the preparation
$P$. Transformations are described by a \emph{transition matrix}
$\Gamma(\lambda'|\lambda,T)$, which is the probability of preparing a
new state $\lambda'$ given the previous state $\lambda$ and the choice
$T$ of transformation. Finally, measurements are represented by a
\emph{response function} $\xi(k|\lambda,M)$ which describes the
probability of an outcome $k$ given the ontic state $\lambda$ and the
choice of measurement $M$. We say that an ontological model
successfully reproduces quantum theory if
\begin{align}
  \label{eq:43}
  \int_\Lambda\dd{\lambda'}\int_\Lambda\dd{\lambda}\xi(k|\lambda',M)\Gamma(&\lambda'|\lambda,T)\mu(\lambda|P)\nonumber\\
  &= \Pr_Q(k|M,T,P)\\
  \forall &P\in\mathcal P, T\in\mathcal T, M\in\mathcal M,\nonumber
\end{align}
or, in a prepare-and-measure-once experiment,
\begin{align}
  \label{eq:44}
  \int_\Lambda\dd{\lambda}\xi(k|\lambda,M)\mu(\lambda|P) &= \Pr_Q(k|M,P)\\
  \forall &P\in\mathcal P, M\in\mathcal M.\nonumber
\end{align}
In both of these cases, $\Pr_Q$ indicates the outcome probability
calculated by operational quantum theory for the particular
experiment under consideration.

Again, we must supplement this definition in order to model sequential
measurement. In textbook quantum theory, the state updates during a
measurement in a way that depends on the previous state, the
measurement procedure, and the measurement outcome: For a measurement
outcome corresponding to a projector $\Pi_k$, a quantum state
$\ket\psi$ after measurement will be
\begin{equation}
  \label{eq:23}
  \ket{\psi'}= \frac{\Pi_k\ket\psi}{\ev{\Pi_k}{\psi}}.
\end{equation}
We allow for dependence on all of these things by choosing to
represent this via a \emph{state update rule}
$\eta(\lambda'|k,\lambda,M)$. As with $\mu$, $\Gamma$, and $\xi$, we
require that $\eta$ is normalized. Although this object looks very
similar to the transition matrix for transformations, it is
distinguished by two important features which we emphasize by choosing
a new symbol to represent it.

The first distinction is simple, in that $\eta$ depends on a
measurement outcome $k$, while $\Gamma$ does not; this is analogous to
the fact that generally in quantum theory we can only implement
measurement update maps probabilistically (i.e. by post-selecting on a
not-necessarily-deterministic measurement outcome).

The second distinction is the fact that $\eta(\lambda'|k,\lambda,M)$
is not defined for all $\lambda\in\Lambda$. Roughly speaking, it
doesn't make sense to ask ``What is the new state $\lambda'$ after
measuring state $\lambda$ and obtaining outcome $k$?'' if the outcome
$k$ could not have occurred given the previous state $\lambda$. To
express this formally, we define the \emph{support} of a distribution
\begin{equation}
  \label{eq:42}
  \supp(\xi(k|\vdot))=\{\lambda\in\Lambda:\xi(k|\lambda)>0\}
\end{equation}
as the set of ontic states on which it is nonzero. Using this, we can
say that $\eta(\lambda'|k,\lambda,M)$ is well-defined only for
$\lambda\in\supp(\xi(k|\vdot))$. This property of $\eta$ is analogous
to the fact that Eq.~\ref{eq:23} is only well-defined for
$\ev{\Pi_k}{\psi}\neq 0$; i.e. only when the measurement $\Pi_k$ could
have responded to the previous quantum state $\psi$.

This second distinction is central to the main result of this paper;
by showing that $\eta$ is non-normalizable on some domain, we are able
to conclude that that domain cannot be part of the support of $\xi$.

The consistency condition with quantum theory is given by equations
similar to Eqs.~\ref{eq:43} and~\ref{eq:44}; see
Appendix~\ref{sec:hmm-appendix} for more formal treatments of the
properties of $\eta$ and other extensions of the ontological models
formalism discussed so far. The rigorous treatment in the appendix
shows that, given the pre-existing assumptions of the ontological
models formalism, this is the only way to represent state update under
sequential measurement; this is why we refer to our results as using
`no additional assumptions.'

This paper is not explicitly concerned with
contextuality~\cite{KochenProblemHiddenVariables1967,SpekkensContextualitypreparationstransformations2005},
but we are careful to ensure that we do not assume
noncontextuality. Under the definition of generalized contextuality
given in~\cite{SpekkensContextualitypreparationstransformations2005},
an ontological model is \emph{noncontextual} if two operationally
equivalent preparation, transformation, or measurement procedures are
always represented by equivalent preparation distributions, transition
matrices, or response functions, respectively\footnote{This definition
  is actually complicated slightly by the inclusion of sequential
  measurements, but we do not discuss this here.}. It is
\emph{contextual} otherwise. While this generalized definition may be
too broad, accounting for contextuality under this definition also
accounts for the traditional
definition~\cite{KochenProblemHiddenVariables1967}, and is therefore
more inclusive. Although, as stated above, we do not assume
noncontextuality of any kind, most of our results hold for a projector
independent of its full measurement context. We are explicit about
this when it is the case, and write $\xi(\Pi|\lambda)$ rather than
$\xi(k=0|\lambda,M=\{\Pi,\ldots\})$ for notational convenience;
$\eta(\lambda'|\lambda,\Pi)$ is defined similarly.

\subsection{$\psi$-epistemic models}
\label{sec:psi-epistemic-models}

Here we focus on a set of precise criteria for $\psi$-epistemic
interpretations within the ontological models formalism. Consider
first the $\psi$-epistemic criterion, proposed
in~\cite{HarriganEinsteinincompletenessepistemic2010} as a test for
whether an interpretation admits at least some quantum states that are
not uniquely determined by the underlying state of
reality. Following~\cite{KaranjaiContextualityboundsefficiency2018},
we account for potential preparation contextuality by defining
\begin{equation}
  \label{eq:36}
  \Delta_\psi =\bigcup_{P_\psi\in\mathcal P_\psi}\supp(\mu(\vdot|P_\psi))
\end{equation}
where $\mathcal P_\psi\subseteq \mathcal P$ is the set consisting of
every possible preparation of $\ket\psi$. We refer to $\Delta_\psi$ as
the support of a state $\ket\psi$, to distinguish it from the support
of a particular preparation $P_\psi$. Then a pair of states
$\ket\psi,\ket\phi$ is \emph{ontologically distinct} in a particular
model if $\Delta_\phi\cap\Delta_\psi=\emptyset$, and
\emph{ontologically indistinct} otherwise\footnote{For this purposes
  of this paper we assume there exists a measure that is absolutely
  continuous with respect to all other measures in the ontological
  model. Therefore, we can work with probability densities, rather
  than the full measure-theoretic treatment. While this assumption is
  not strictly true in all of our models, it does not affect our
  results and significantly simplifies our presentation.}. This leads
us to the standard definition of a $\psi$-epistemic ontological
model~\cite{HarriganEinsteinincompletenessepistemic2010,Leiferquantumstatereal2014}:
\begin{definition}[$\psi$-epistemic]
  \label{def:psi-epistemic}
  An ontological model is $\psi$-epistemic if there exists a pair of
  states $\ket\psi,\ket\phi$ that are ontologically indistinct;
  i.e.
  \begin{equation}
  \exists \ket\psi,\ket\phi:\Delta_\phi\cap\Delta_\psi\neq\emptyset.
\end{equation}

\end{definition}
As noted in~\cite{Leiferquantumstatereal2014}, this definition is
highly permissive in the sense that, if an ontological model were to
contain only a single pair of quantum states that are ontologically
indistinct, then it would not achieve the full explanatory power
expected of the $\psi$-epistemic viewpoint; this is exactly the case
with the \abclz model discussed in Section~\ref{sec:abcl-model}.

There are, however, proposals to strengthen the notion of
$\psi$-epistemicity, two of which are relevant to our
discussion~\cite{Leiferquantumstatereal2014,MontinaCommentmathoverflow2012}.
\begin{definition}[Pairwise $\psi$-epistemic]
  \label{def:pairwise}
  An ontological model is \emph{pairwise $\psi$-epistemic} if, for all
  pairs $\ket\psi,\ket\phi$ of nonorthogonal quantum states,
  $\ket\psi$ and $\ket\phi$ are ontologically indistinct.
\end{definition}
\begin{definition}[Never $\psi$-ontic]
  \label{def:never-ontic}
  An ontological model is \emph{never $\psi$-ontic} if every ontic
  state $\lambda\in\Lambda$ is in the support of at least two quantum
  states:
  \begin{equation}
    \label{eq:11}
    \forall\lambda\in\Lambda: \exists\psi,\phi: \lambda\in\Delta_\psi\cap\Delta_\phi.
  \end{equation}
\end{definition}
Note that both of these definitions imply the weaker notion of
$\psi$-epistemicity, but are independent from one another.

\subsection{Some easy cases of state update rules}
\label{sec:some-easy-cases}
There are two cases in which, given a prepare-and-measure-once
ontological model for quantum theory, we can always augment it with a
state update rule for measurement. First, if we only include rank-1
projective measurements in the subtheory we're modeling, we can simply
re-prepare in the measured (unique, pure) state:
\begin{equation}
  \label{eq:19}
  \eta(\lambda'|k,\lambda,M_{\{\Pi_i\}}) = \mu(\lambda'|P_{\Pi_k}) \text{ for }\tr(\Pi_k)=1.
\end{equation}
This is normalized for all $\lambda$ since $\mu$ is normalized, and
faithfully reproduces quantum statistics since $\mu$ does. It is
independent of the previous state $\lambda$, which, besides being
unsatisfying, is also not possible in general (this follows from
Section~\ref{sec:main-thm}).

Second, it is quick to prove, again by construction, that $\psi$-ontic
models can always be given a state update rule. Since there is a
unique quantum state $\ket{\psi_\lambda}$ associated with every ontic
state $\lambda$, we can define
\begin{equation}
  \label{eq:20}
  \eta(\lambda'|k,\lambda,M) = \mu\qty(\lambda'\bigg|P=\frac{M_k\proj{\psi_\lambda}M_k^\dag}{\tr(M_k\proj{\psi_\lambda}M_k^\dag)}).
\end{equation}
Again, normalization and faithfulness follow because $\mu$ has these
properties. Note that this construction works for any kind of
measurement, not just projective measurements.

These observations together suggest that in order to find anything
interesting involving state update, we ought to examine higher-rank
measurements in $\psi$-epistemic models. This suspicion will be
confirmed by the main result of this paper, which applies to exactly
these types of measurements and models.

\section{The consequences of state update}

\label{sec:main-thm}
We now prove our central claim that consideration of a rule for state
update under measurement has consequences for the response function of
an ontological model, so that consistent state update puts
restrictions on how one may represent even a prepare-and-measure-once
experiment. We begin with a lemma that articulates a general property
of the update rule $\eta$, and then examine its consequences for
response functions $\xi$.

\begin{lemma}
  \label{lem:littlelemma}
  Suppose we have an ontological model with ontic space $\Lambda$,
  preparation distributions $\mu(\lambda|P)$, indicator functions
  $\xi(k|\lambda,M)$, and state update maps
  $\eta(\lambda'|k,\lambda,M)$. For a particular ontic state $\lambda$
  and measurement projector $\Pi$, we define the set $S_{\lambda,\Pi}$
  of quantum states that one could obtain after measurement of any
  quantum state consistent with $\lambda$:
  \begin{equation}
    \label{eq:32}
    S_{\lambda,\Pi}=\left\{\frac{\Pi\ket\phi}{\sqrt{\ev{\Pi}{\phi}}}\:\middle|\:\forall\ket\phi:\lambda\in\Delta_\phi\right\}.
  \end{equation}
  It is then true that, independently of the measurement context of
  $\Pi$,
  \begin{equation}
    \label{eq:2}
    \supp(\eta(\vdot|\lambda,\Pi)) \subseteq \bigcap_{\ket\psi\in S_{\lambda,\Pi}}\Delta_\psi.
  \end{equation}
\end{lemma}
\begin{proof}
  Suppose that measuring a state $\ket\phi$ with a measurement $M$
  results in the updated state $\ket\psi$ when we get outcome $k$,
  where $\Pi$ is the $k$th projector in $M$. Then let
  $P_{M,k,P_\phi}\in\mathcal P_\psi$ be the preparation procedure
  associated with post-selection of this measurement outcome after a
  particular preparation $P_\phi$. It must be normalized on
  $\Delta_\psi$:
  \begin{align}
    1 &= \int_{\Delta_\psi}\dd{\lambda'}\mu(\lambda'|P_{M,k,P_\phi})\nonumber\\
      &= \int_{\Delta_\psi}\dd{\lambda'}\int_{\Delta_\phi}\dd{\lambda}\eta(\lambda'|k,\lambda,M)\mu(\lambda|P_\phi)\nonumber\\
      &= \int_{\Delta_\phi}\dd{\lambda}\mu(\lambda|P_\phi)\int_{\Delta_\psi}\dd{\lambda'}\eta(\lambda'|k,\lambda,M)\nonumber
  \end{align}
  Since $\eta$ is always positive, normalization of
  $\mu(\lambda|P_\phi)$ then implies that
  \begin{equation}
    \label{eq:6}
    \int_{\Delta_\psi}\dd{\lambda'}\eta(\lambda'|k,\lambda,M) = 1 \qquad\forall\lambda\in\Delta_{\phi}.\nonumber
  \end{equation}
  If $\eta$ is normalized on a region, its support must be contained
  in that region. Thus for all $\lambda$ that are consistent with some
  preparation $\ket\phi$ that could result in the post-measurement
  state $\ket\psi$,
  \begin{equation}
    \label{eq:7}
    \supp(\eta(\vdot|k,\lambda,M)) \subseteq \Delta_\psi.\nonumber
  \end{equation}
  The fact that this is true for all $\ket\psi$ that could result from the
  measurement leads to Eq.~\ref{eq:2}.
\end{proof}
 
We note that Lemma~\ref{lem:littlelemma} is easy to account for in
$\psi$-ontic theories and for rank-1 measurements, since in both cases
$S_{\lambda,\Pi}$ has a single element. This is why we were able to
write down update rules for these situations in
Section~\ref{sec:some-easy-cases}.  Outside of these trivial cases,
Eq.~(\ref{eq:2}) is a very restrictive condition; depending on the
structure of $S_{\lambda,\Pi}$, the intersection may be a very small
set. In particular, if any two of the post-selected quantum states are
orthogonal, then $S_{\lambda,\Pi}$ is empty.

\begin{theorem}[Main theorem]
  \label{thm:main-thm}
  Suppose that a projector $\Pi$ maps two states
  $\ket\alpha,\ket\beta$ to ontologically distinct states
  $\Pi\ket\alpha,\Pi\ket\beta$. Then the response function for $\Pi$
  cannot have support on the overlap of $\ket\alpha,\ket\beta$ for any
  measurement context of $\Pi$; i.e.
  \begin{equation}
    \label{eq:3}
    \Delta_{\Pi\ket\alpha}\cap\Delta_{\Pi\ket\beta}=\emptyset
    \implies \xi(\Pi|\lambda) = 0 \ \ \forall \lambda\in\Delta_\alpha\cap\Delta_\beta.
  \end{equation}
\end{theorem}
\begin{proof}
  Pick some $\lambda\in\Delta_\alpha\cap\Delta_\beta$. By
  Lemma~\ref{lem:littlelemma},
  $\supp(\eta(\vdot|\lambda,\Pi))=\emptyset$ so
  $\eta(\lambda'|\lambda,\Pi)$ is not normalizable. As discussed in
  Section~\ref{sec:ont-models}, this is only allowable if
  $\xi(\Pi|\lambda)=0$.
\end{proof}

Both the lemma and the theorem hold for non-projective measurements as
well. We emphasize that this result does not say anything
\emph{directly} about the overlap of the supports of quantum states,
just their overlap within the support of a particular response
function. In the following section, we deploy this theorem by showing 
that, in every known $\psi$-epistemic model for $d\geq3$, measurements
of this type exist \emph{and} have support on the relevant overlaps,
leading to contradiction and demonstrating that these models cannot
reproduce state update under measurement.

\section{Examples of state update under measurement (or its impossibility)}
\label{sec:exampl-ontol-models}
We provide a number of examples of ontological models from the
literature, illustrating some of the properties described in the
previous sections. For each model, we either specify its state update
rule or prove that it cannot reproduce state update. Although many of
these models can be easily defined for arbitrary types of
measurements, we only consider projective measurements for simplicity
and notational consistency.

\begin{figure*}[t]
  \centering
 \begin{subfigure}[t]{0.24\textwidth}
 \includegraphics[width=\textwidth]{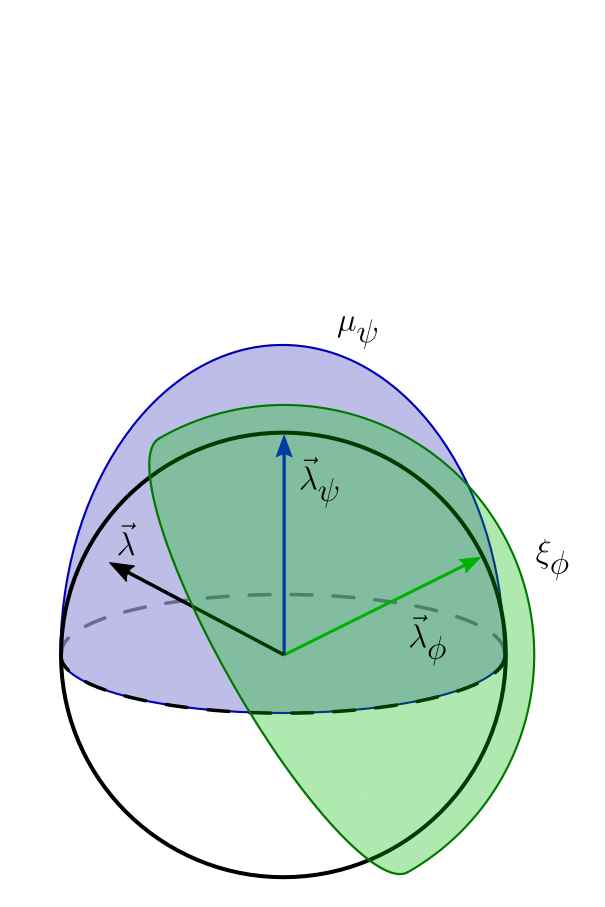}
 \caption{Kochen-Specker model}
 \label{fig:KSmodel}
 \end{subfigure}
 \begin{subfigure}[t]{0.24\textwidth}
 \includegraphics[width=\textwidth]{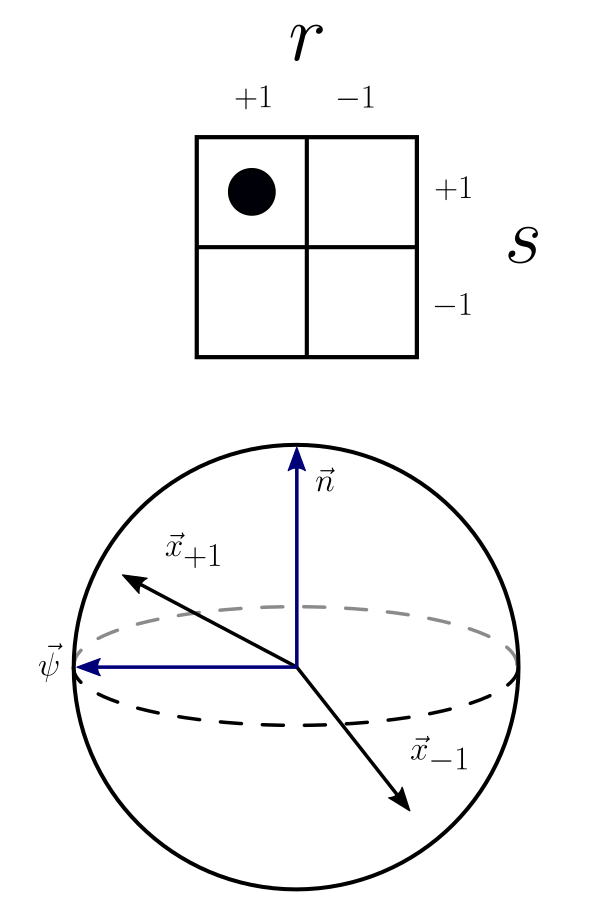}
 \caption{Montina model}
 \label{fig:Montmodel}
 \end{subfigure}
 \begin{subfigure}[t]{0.24\textwidth}
 \includegraphics[width=\textwidth]{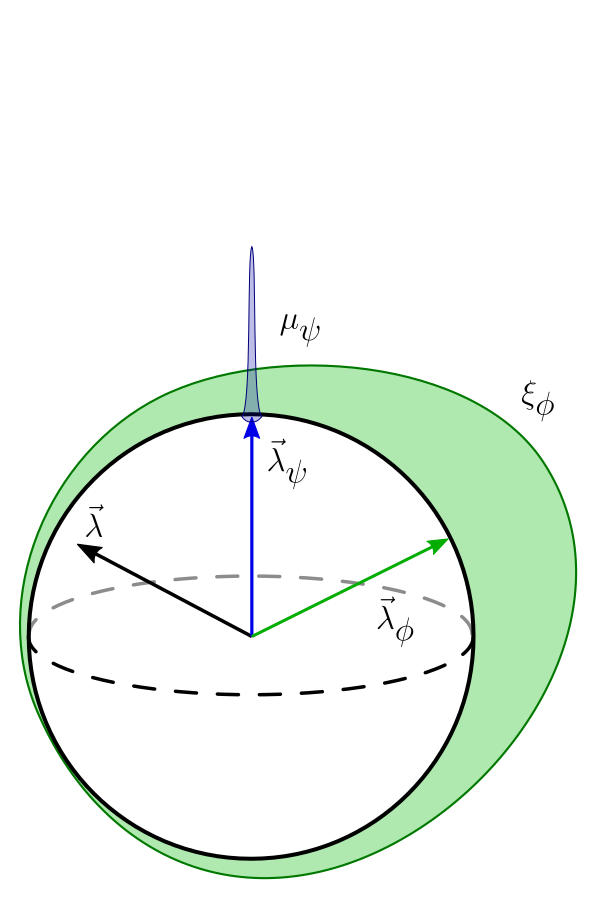}
 \caption{Beltrametti-Bugajski model}
 \label{fig:BBmodel}
 \end{subfigure}
 \begin{subfigure}[t]{0.24\textwidth}
 \includegraphics[width=\textwidth]{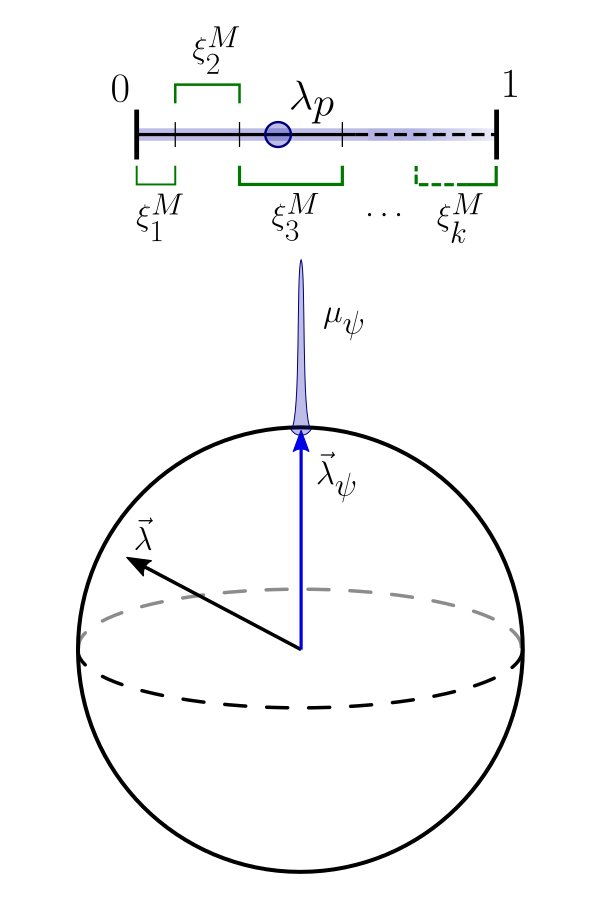}
 \caption{Bell Model}
 \label{fig:Bellmodel}
 \end{subfigure}
 \caption{Visualizations of the state space, preparation
   distributions, and response functions for (a) the Kochen-Specker
   model, (b) Montina's model, (c) the Beltrametti-Bugajski model for
   $d=2$, and (d) Bell's model for $d=2$. Blue represents the support
   of preparations, and green the support of the response functions,
   where possible. Black objects are generic elements of the state
   space.}
  \label{fig:models}
\end{figure*} 

\subsection{$\psi$-epistemic models of a qubit}
The only $\psi$-epistemic models that are able to represent state
update are those that restrict to modeling a qubit. This corresponds
to the fact that the only nontrivial projective measurements in a
qubit are rank-1, which, as described earlier, do not result in
restrictions from state update.

\subsubsection{Kochen-Specker model}
\label{sec:kochen-specker-model}
The Kochen-Specker model of a
qubit~\cite{KochenProblemHiddenVariables1967,Leiferquantumstatereal2014}
is an exemplar of what we look for in a $\psi$-epistemic theory, with
the unfortunate feature that it only works in $d=2$ dimensions. It is
both pairwise $\psi$-epistemic and never $\psi$-ontic, and provides a
very intuitively pleasing interpretation of the statistical nature of
quantum theory. We take the ontic space to be the unit sphere $S^2$,
and denote by $\vec\psi$ the Bloch vector corresponding to $\ket\psi$
under the usual mapping. Preparations and measurement outcomes are
represented by distributions over hemispheres, with response functions
uniform and preparation distributions peaked towards the center
(Fig.~\ref{fig:KSmodel}). Unitary transformations are represented by
rotations of the sphere. Since the only nontrivial measurements on a
qubit are rank-1 measurements, this is a case where we can use the
state update rule described in Eq.~\ref{eq:19} and just re-prepare the
measured state for our update rule:
\begin{align}
  \Lambda
  &= S^2\nonumber\\
  \mu(\vec\lambda|P_\psi)
  &=\frac{1}{\pi}\Theta(\vec\psi\cdot\vec\lambda)\vec\psi\cdot\vec\lambda\nonumber\\
  \Gamma(\vec\lambda'|\vec\lambda,T_U)
  &= \delta(\vec\lambda'-R_U\vec\lambda)\nonumber\\
  \xi(k|\vec\lambda,M_\phi)
  &= \Theta(k\vec\phi\cdot\vec\lambda)\nonumber\\
  \eta(\vec\lambda'|k,\vec\lambda,M_\phi)
  &= \frac{1}{\pi}\Theta(k\vec\phi\cdot\vec\lambda')k\vec\phi\cdot\vec\lambda'
\end{align}
Here $\Theta$ is the Heaviside step function, $R_U$ is the rotation of
the Bloch sphere corresponding to a unitary $U$, and
$k\in\{+1,-1\}$. This particular state update rule is not particularly
satisfactory in an explanatory sense, since it is independent of the
previous ontic state.

\subsubsection{Montina model}
\label{sec:montina-model}
In~\cite{MontinaDynamicsqubitclassical2012}, Montina introduces an
ontological model based on the Kochen-Specker model. The model was
constructed to show that state update in a qubit can be successfully
modeled by only updating a finite amount of information in the ontic
state. To do so, Montina extends the ontic space of the Kochen-Specker
model by taking two vectors $\vec x_{+1},\vec x_{-1}$ on the Bloch
sphere and adding two bits, labeled $r$ and $s$, such that the vectors
on the Bloch sphere are dynamic under transformations but remain
static under state update (Fig.~\ref{fig:Montmodel}).  The bit
$r\in\{-1,+1\}$ acts as an index which decides which of the two Bloch
vectors is `active;' $s\in\{-1,+1\}$ stores the result of a
hypothetical standard basis measurement on the state. We take the
standard basis to be defined by a special vector $\vec n$ pointing
along the $z$-axis.

As in the Kochen-Specker model, unitary transformations act by
rotating the Bloch vectors; additionally, if the vector $\vec x_r$
(i.e. the `active' Bloch vector) crosses the horizontal equator of the
sphere during this transformation, then the bit $s$ flips to $-s$. $r$
does not change during a transformation.

A measurement in the standard basis simply reveals the value of $s$,
and then updates $r$ so that the active vector is the one which was
more closely aligned with $\vec n$ at the time of measurement. For any
other basis, we apply the unitary that maps our desired measurement
basis to the standard basis, measure, and then rotate back---this
whole process has been wrapped into our definitions of $\eta$ and
$\xi$ below. In either case, the vectors $\vec x_{+1},\vec x_{-1}$ do not
change during measurement.

Finally, we prepare a state $\vec\psi$ by measuring in the basis
$\{\vec\psi,-\vec\psi\}$ and applying a rotation that maps
$-\vec\psi\to\vec\psi$ if we measured $-\vec\psi$. Summarizing these
constructions, we can write
\begin{align}
  \Lambda &= S^2\times S^2\times\{-1,+1\}\times\{-1,+1\}\nonumber\\
  \lambda &= (\vec x_{+1},\vec x_{-1}, r,s)\nonumber\\
  \mu(\lambda|P_\psi)
          &= \frac{1}{(4\pi)^2}  \Theta\left[ s\left(\vec x_{r}\cdot\vec\psi \right)\left(\vec x_{r}\cdot\vec n \right) \right]\nonumber\\
          &\quad\cdot\Theta\left[ r\left[
            \left( \vec x_{+1}\cdot\vec \psi \right)^2
            -\left( \vec x_{-1} \cdot\vec\psi \right)^2 \right] \right]\nonumber\\
  \Gamma(\lambda'|\lambda,T_U)
          &=
            \delta\left(\vec x_{+1}'-R_U\vec x_{+1}\right)
            \delta\left(\vec x_{-1}'-R_U\vec x_{-1}\right)\nonumber\\
          &\quad\cdot\Theta\left[ ss'\left( \vec x_r'\cdot \vec n\right)\left( \vec x_r\cdot \vec n \right) \right]
            \Theta[rr'] \nonumber\\
  \xi(k|\lambda,M_\phi) 
          &= \Theta\left[ ks\left( \vec x_r\cdot\vec n \right)
            \left( \vec x_r\cdot\vec \phi \right)\right]\nonumber\\
  \eta(\lambda'|k,\lambda,M_\phi)
          &= \delta\left(\vec x_{+1}'-\vec x_{+1}\right)
            \delta\left(\vec x_{-1}'-\vec x_{-1}\right)\nonumber\\
          &\quad\cdot\Theta\left[ ss'\qty( \vec x_r\cdot \vec n)\qty( \vec x_r\cdot \vec\phi)
            \qty( \vec x_{r'}\cdot \vec n)\qty( \vec x_{r'}\cdot \vec\phi)\right]\nonumber\\
          &\quad\cdot\Theta\left[ r'\left[
            \left( \vec x_{+1}\cdot\vec \phi \right)^2
            -\left( \vec x_{-1} \cdot\vec\phi \right)^2 \right] \right]
\end{align}
The original presentation is not stated in terms of the ontological
models formalism, and only explicitly models measurements in the
standard basis. This lead to the claim that state update under
measurement is accounted for by updating a single bit, but it is clear
from the form of $\eta$ above that by including all measurements in
our subtheory we have caused both bits to be updated during
measurement.

We include this model here for two reasons. First, it is one of the
few ontological models in the literature that has explicitly
considered state update under measurement. Second, it demonstrates
that the generic rank-1 update (Eq.~\ref{eq:19}) that we used for the
Kochen-Specker model is not the only possibility; even though all
measurements in this model are rank-1, $\eta$ has nontrivial
dependence on the previous ontic state $\lambda$. Thus just because we
\emph{can} construct a trivial update rule in some cases does not mean
that there is then nothing interesting to investigate. It also
includes these features while remaining pairwise $\psi$-epistemic and
never $\psi$-ontic.

\subsection{Models of full quantum theory for arbitrary dimension}

\subsubsection{Beltrametti-Bugajski model}
\label{sec:beltr-bugajski-model}
The Beltrametti-Bugajski
model~\cite{Beltrametticlassicalextensionquantum1995,Leiferquantumstatereal2014}
is perhaps the simplest ontological model that describes a system of
arbitrary dimension. Although it is $\psi$-ontic, it is the starting
point for the construction of the next three models in this
section. For a $d$-dimensional quantum system, we take the ontic space
to be the quantum state space, which we denote $\mathcal{PH}^{d-1}$
(the projective Hilbert space of dimension $d-1$). Preparations,
transformations, measurements, and state update rules then follow
directly from the usual quantum rules:
\begin{align}
  \Lambda
  &= \mathcal{PH}^{d-1}\nonumber\\
  \mu(\lambda|P_\psi)
  &=\delta(\ket\lambda-\ket\psi)\nonumber\\
  \Gamma(\lambda'|\lambda,T_U)
  &= \delta(\ket{\lambda'}-U\ket\lambda)\nonumber\\
  \xi(k|\lambda,M_{\{\Pi_i\}})
  &= \bra\lambda \Pi_k\ket\lambda\nonumber\\
  \eta(\lambda'|k,\lambda,M_{\{\Pi_i\}})
  &= \delta\qty(\ket{\lambda'}-\frac{\Pi_k\ket\lambda}
    {\sqrt{\bra\lambda \Pi_k\ket\lambda}})
\end{align}
This provides an example of the generic update-rule for $\psi$-ontic
models (Eq.~\ref{eq:20}), and is depicted in Fig.~\ref{fig:BBmodel}.

\subsubsection{Bell's model}
\label{sec:bells-model}
Lewis et al.~\cite{LewisDistinctQuantumStates2012} extended a model of
a qubit orignally proposed by
Bell~\cite{BellProblemHiddenVariables1966} to arbitrary dimension,
which can be seen as a modification of the Beltrametti-Bugajski
model~\cite{Leiferquantumstatereal2014}. The ontic space is the
Cartesian product of the projective Hilbert space with the unit
interval $[0,1]$. Now we write $\lambda$ as an ordered pair
$\lambda= (\ket\lambda,p_\lambda)$ where
$\ket\lambda\in\mathcal{PH}^{d-1}$, as in the Beltrametti-Bugajski
model, and $p_\lambda\in[0,1]$. Preparations remain essentially the
same, becoming a product distribution of a delta function on the
quantum state space with a uniform distribution over the unit
interval. The response functions divide up the unit interval into
lengths corresponding to probabilities of measuring each outcome, and
respond with outcome $k$ when $p_\lambda$ is in the corresponding
interval (Fig.~\ref{fig:Bellmodel}). This has the effect of making the
model outcome deterministic.
\begin{align}
  \Lambda&= \mathcal{PH}^{d-1}\times[0,1]\nonumber\\
  \mu(\lambda|P_\psi) &= \delta(\ket\lambda-\ket\psi)\nonumber\\
  \Gamma(\lambda'|\lambda,T_U)
         &= \delta(\ket{\lambda'}-U\ket\lambda)\nonumber\\
  \xi(k|\lambda,M_{\{\Pi_i\}})&= \Theta\qty[p_\lambda-\sum_{j=0}^{k-1}\tr(\Pi_j\proj\lambda)]\nonumber\\
         &\quad\cdot\Theta\qty[-p_\lambda+\sum_{j=0}^{k}\tr(\Pi_j\proj\lambda)]\nonumber\\
  \eta(\lambda'|k,\lambda,M_{\{\Pi_i\}})
         &= \delta\qty(\ket{\lambda'}-\frac{\Pi_k\ket\lambda}
           {\sqrt{\bra\lambda \Pi_k\ket\lambda}})
\end{align}
Since this model is still $\psi$-ontic, we once again use the generic
state update rule for $\psi$-ontic models. In this case, we can also
see that this works because of the structure of the preparations as
product distributions. Since every state has a uniform distribution
over $p_\lambda$, and this is uncorrelated with $\ket\lambda$, we can
just update $\ket\lambda$ according to the Beltrametti-Bugajski update
rule and then re-randomize uniformly over $p_\lambda$.

\subsubsection{LJBR model}
\label{sec:ljbr-model}

In~\cite{LewisDistinctQuantumStates2012}, Lewis et al. define a
$\psi$-epistemic model based on their generalization of Bell's model.
Referred to here as the LJBR model, it is motivated by the observation
that the order of segments in the response function of Bell's model
does not matter: a re-ordering of these segments allows arbitrary
modification of preparation distributions within a subset of the ontic
space, so they can be made to overlap. We present here a brief
description of the `most epistemic' version of this model, and refer
the reader to~\cite{LewisDistinctQuantumStates2012} for a more
thorough construction and motivation.

The LJBR model has the same ontic space as the Bell model, so we again
write ontic states as $\lambda=(\ket\lambda,p_\lambda)$. It is
constructed in a preferred basis $\{\ket j\}$, which we use in
defining two helper functions. First,
\begin{equation}
  \label{eq:17}
  z_j(\ket\lambda) = \inf_{\ket\phi:\tr(\proj j\proj\phi)\geq 1/d} \tr(\proj\lambda\proj\phi).
\end{equation}
Note that $z_j(\ket\lambda)>0$ if and only if
$\tr(\proj j\proj\lambda)>\frac{d-1}{d}$, so $z_j(\ket\lambda)$ is
nonzero for at most a single element of the preferred basis; we denote
this unique vector as $\ket{j_\lambda}$. Second, we define a
permutation $\pi_{M,\lambda}$ for each measurement $M$ and ontic state
$\lambda$:
\begin{align}
  \label{eq:28}
  \tr(M_{\pi_{M,\lambda}(0)}\proj{j_\lambda})&\geq \tr(M_{\pi_{M,\lambda}(1)}\proj{j_\lambda})\geq\cdots\nonumber\\
  \cdots&\geq \tr(M_{\pi_{M,\lambda}(\qty|M|-1)}\proj{j_\lambda}).
\end{align}
If there is no $\ket{j_\lambda}$, i.e. $z_j(\ket\lambda)=0$ for all
$j$, then we take $\pi_{M,\lambda}$ to be the identity
permutation. The final element we need before defining the model
itself is a set
\begin{equation}
  \label{eq:31}
  \mathcal E_j = \{\lambda:  z_j(\ket\lambda) > 0\}
\end{equation}
defined for each basis vector. Without further ado, the full
specification of the model:
\begin{align}
  \Lambda&= \mathcal{PH}^{d-1}\times[0,1]\nonumber\\
  \lambda&= (\ket\lambda,p_\lambda)\nonumber\\
  \mu(\lambda|P_\psi)&= \delta(\ket\lambda-\ket\psi)\prod_j\Theta[p_\lambda-z_j(\ket\psi)]\nonumber\\
  &\quad+ \sum_j z_j(\ket\psi)\mu_{\mathcal E_j}(\lambda) \nonumber\\
  \xi(k|\lambda,M_{\{\Pi_i\}})&= \Theta\qty[p_\lambda-\sum_{l=0}^{k-1}\tr(\Pi_{\pi_{M,\lambda}(l)}\proj\lambda)]\nonumber\\
  &\quad\cdot\Theta\qty[-p_\lambda+\sum_{l=0}^{k}\tr(\Pi_{\pi_{M,\lambda}(l)}\proj\lambda)]  
\end{align}
where $\mu_{\mathcal E_j}(\lambda)$ is the uniform distribution over
$\mathcal E_j$.

Roughly, all quantum states $\ket\psi$ with
$\tr(\proj j\proj \psi)>\frac{d-1}{d}$ will have support on
$\mathcal E_j$, and so will all overlap with each other. The
permutation included in the definition of the measurements is
constructed so that this shared support does not affect the
prepare-and-measure-once statistics: the measurement ordered first by
the permutation has a support which entirely contains $\mathcal
E_j$. This model is not pairwise $\psi$-epistemic, nor is it never
$\psi$-ontic. Note additionally that it was originally only defined
for rank-1 projective measurements, but it works just as well for
higher-rank projective measurements without modification.

This model is the first to fall to Theorem~\ref{thm:main-thm}:
\begin{theorem}
  The LJBR model cannot represent state update under measurement in
  dimension $d\geq3$.
\end{theorem}
\begin{proof}
  The general idea of the proof is to find a measurement which maps
  any two nonidentical states $\ket\alpha,\ket\beta$ to two again
  nonidentical states $\Pi\ket\alpha,\Pi\ket\beta$ which both have no
  support on any of the $\mathcal E_j$, and so are ontologically
  distinct.

  Consider the preferred basis $\ket j$ of the LJBR model and the
  generalized $x$-basis defined by
  \begin{equation}
    \label{eq:12}
    \ket{X_k} = \frac{1}{\sqrt{d}}\sum_{j=0}^{d-1} \omega^{jk}\ket j,\qquad \omega=e^{2\pi i /d}.
  \end{equation}
  These $x$-basis states have the property
  $\tr(\proj{X_j}\proj k)=\frac{1}{d}$ for all $j,k$. There must exist
  two elements $\ket{X_{k_1}},\ket{X_{k_2}}$ of the $x$-basis such
  that $\ket\alpha,\ket\beta$ differ on that two-dimensional subspace
  or else $\ket\alpha,\ket\beta$ would be identical. Pick two such
  elements, and consider the projector
  $\Pi=\proj{X_{k_1}}+\proj{X_{k_2}}$. The quantum overlap of the
  post-measurement state $\Pi\ket\alpha$ with any basis vector
  $\ket j$ is, using the submultiplicativity of the trace,
  \begin{equation}
    \label{eq:13}
    \frac{\tr(\proj j\Pi\proj\alpha\Pi)}{\tr(\Pi\proj\alpha)}
    \leq\tr(\proj j\Pi) = \frac{2}{d} \leq \frac{d-1}{d}
  \end{equation}
  and the same is true for $\Pi\ket\beta$. As described above, only
  states with $\tr(\proj j\proj\psi)>\frac{d-1}{d}$ have overlap with
  any other states in the LJBR model, so the post-measurement states
  are ontologically distinct; by Theorem~\ref{thm:main-thm},
  $\xi(\Pi|\lambda) =0$ for all $\lambda\in\mathcal E_j$
  for all $j$ since $\ket\alpha$ and $\ket\beta$ were arbitrary.

  However, when measured in the context of the rest of the rank-1
  $x$-basis projectors, $\Pi$ will be ordered first by
  $\pi_{M,\lambda}$ for all $\lambda$ since
  \begin{equation}
    \label{eq:16}
    \tr(\Pi\proj j)=\frac{2}{d}>\frac{1}{d}
  \end{equation}
  for all $j$. Thus, by the construction of the LJBR model,
  $\xi(\Pi|\lambda)=1$ for all $\lambda\in\mathcal E_j$, resulting in
  a contradiction.
\end{proof}

\subsubsection{ABCL models}
\label{sec:abcl-model}
In~\cite{Aaronsonpsepistemictheoriesrole2013}, Aaronson et. al. construct
two $\psi$-epistemic models. The first, which we will call \abclz, is
very closely related to the LJBR model but is not identical; rather
than continuous regions of quantum states which overlap, this model
has exactly one pair of quantum states which are ontologically
indistinct. However, it gains the feature that \emph{any} two
nonorthogonal quantum states can be chosen as the single pair that
overlaps. With malice aforethought, we will call this defining pair
$\ket\alpha,\ket\beta$.

The second, \abclo, is a convex mixture (to be defined) of the \abclz
model constructed for all $\ket\alpha,\ket\beta$ and is intended to
demonstrate the possibility of a pairwise $\psi$-epistemic model. This
is the only known example of a pairwise $\psi$-epistemic model in
$d\geq3$, but it still is not never
$\psi$-ontic~\cite{Leiferquantumstatereal2014,MontinaCommentmathoverflow2012}. These
models have come under criticism for their ``unnaturalness,'' but we
show here that their problems go deeper due to an inability to
represent state-update.

We begin with \abclz, defining a couple of helper functions like in
the LJBR model. Rather than ordering measurements with respect to
traces with a preferred basis, we use the defining states
$\ket\alpha,\ket\beta$ and a function
\begin{equation}
  \label{eq:18}
  g_{\alpha\beta}(\Pi)=\min\{\tr(\Pi\proj\alpha),\tr(\Pi\proj\beta)\}.
\end{equation}
We now define a new permutation $\sigma_M$\footnote{To be precise, we
  should label this with $\alpha,\beta$ as well to emphasize that it
  belongs to the model defined by that particular pair of states.}
for each measurement $M$~\cite{Leiferquantumstatereal2014}:
\begin{equation}
  \label{eq:1}
  g_{\alpha\beta}(M_{\sigma_M(0)})\geq g_{\alpha\beta}(M_{\sigma_M(1)})\geq\cdots\geq g_{\alpha\beta}(M_{\sigma_M(|M|-1)}).
\end{equation}
With this, we can specify the \abclz model.
\begin{align}
  &\qquad\Lambda= \mathcal{PH}^{d-1}\times[0,1]\nonumber\\
  &\qquad\lambda= (\ket\lambda,p_\lambda)\nonumber\\
  &\mu(\lambda|P_\psi)
    =
    \begin{cases}                          
      \begin{aligned}
        &\bigg[\Theta(p_\lambda-\varepsilon)\delta(\ket\lambda-\ket\psi)\\
        &+\frac{1}{2}\Theta(\varepsilon-p_\lambda)
        \qty[\delta(\ket\lambda-\ket\alpha)+\delta(\ket\lambda-\ket\beta)]\bigg]
      \end{aligned}\\
      \hfill\text{if }\ket\psi=\ket\alpha,\ket\beta\\\\
      \delta(\ket\lambda-\ket\psi)
      \hfill\text{otherwise}
    \end{cases}\nonumber\\
  &\xi(k|\lambda,M_{\{\Pi_i\}})= \Theta\qty[p_\lambda-\sum_{j=0}^{k-1}\tr(\Pi_{\sigma_M(j)}\proj\lambda)]\nonumber\\
                              &\qquad\qquad\qquad\cdot\Theta\qty[-p_\lambda+\sum_{j=0}^{k}\tr(\Pi_{\sigma_M(j)}\proj\lambda)]
\end{align}
for $\varepsilon\leq\frac{|\braket{\alpha}{\beta}|}{d}$. Now the
preparation distributions for $\ket\alpha$ and $\ket\beta$ overlap on
$\{\ket\alpha,\ket\beta\}\times[0,\varepsilon]$; as in the LJBR model,
the permutation in $\xi$ ensures that the preparation change doesn't
affect the prepare-and-measure-once statistics by making sure
measurements whose support must contain this overlap region are
ordered first. Once again, this model fails to meet the conditions
required in order to faithfully represent state update:
\begin{theorem}
  The \abclz model cannot represent state update under measurement in
  dimension $d\geq3$.
\end{theorem}
\begin{proof}
  Call the two states defining the model $\ket\alpha,\ket\beta$. Let
  $\Pi=\proj\alpha+\proj\gamma$, where $\ket\gamma$ is some state such
  that
  \begin{equation}
    \braket{\alpha}{\gamma}=0 \quad\text{and}\quad
    0<\qty|\braket{\gamma}{\beta}|^2<1-\qty|\braket{\alpha}{\beta}|^2.\label{eq:22}
  \end{equation}
  Under this measurement, $\ket\alpha$ maps to $\ket\alpha$ and
  $\ket\beta$ does not get mapped to either $\ket\alpha$ or
  $\ket\beta$. Thus the post-measurement states are ontologically
  distinct, so by Theorem~\ref{thm:main-thm}, $\xi(\Pi|\lambda)=0$ for
  all $\lambda\in\Delta_\alpha\cap\Delta_\beta$.

  For the other half of the contradiction, note that
  $\tr(\Pi\proj\alpha)=1$ means
  $g_{\alpha\beta}(\Pi)=\tr(\Pi\proj \beta)>0$ (since
  $\ket\alpha,\ket\beta$ are nonorthogonal) and
  $g_{\alpha\beta}(\Id-\Pi)=1-\tr(\Pi\proj \alpha)=0$, so $\Pi$ is
  ordered first by $\sigma_M$ when measured in the context
  $M=\{\Pi,\Id-\Pi\}$. Thus $\xi(\Pi|\lambda)=1$ for all
  $\lambda\in\Delta_\alpha\cap\Delta_\beta$, resulting in a
  contradiction.
\end{proof}

We outline the \abclo model schematically and refer the reader
to~\cite{Leiferquantumstatereal2014,Aaronsonpsepistemictheoriesrole2013}
for details. Given two ontological models specified by
$\Lambda_1,\mu_1,\xi_1$ and $\Lambda_2,\mu_2,\xi_2$ respectively, the
authors define a convex combination of these models as
$\Lambda_3,\mu_3,\xi_3$ such that
\begin{align}
  \label{eq:33}
  \Lambda_3&=\Lambda_1\oplus\Lambda_2\nonumber\\
b  \mu_3&= p\mu_1+(1-p)\mu_2\nonumber\\
  \xi_3&= \xi_1+\xi_2
\end{align}
Here $p\in(0,1)$ is some mixing parameter. If there's overlap between
two states in either of models 1 or 2, then model 3 has overlap on
these states. The \abclo model is then defined essentially as a convex
mixture of the \abclz models for \emph{all} pairs
$\ket\alpha,\ket\beta$, taking care with respect to the uncountable
size of this set.

In order to include state update in a convex combination of
ontological models, the most obvious (and perhaps only) option is to
specify
\begin{equation}
  \eta_3=\eta_1+\eta_2.
\end{equation}
The failure of the \abclo model to reproduce state update follows
directly from the failure of the \abclz model.
\begin{theorem}
  The \abclo model cannot represent state update under measurement in
  dimension $d\geq3$
\end{theorem}
\begin{proof}
  When we take a convex combination of models, we see that
  \begin{align*}
    \supp(\xi_3(k|\vdot,M))           = &\supp(\xi_1(k|\vdot,M)) \cup \supp(\xi_2(k|\vdot,M)) \\
    \supp(\eta_3(\vdot|k,\lambda,M)) & =\supp(\eta_1(\vdot|k,\lambda,M))                      \\
                                     & \qquad\qquad\cup \supp(\eta_2(\vdot|k,\lambda,M))
  \end{align*}
  Thus if either of models 1 or 2 violates Theorem~\ref{thm:main-thm},
  model 3 must violate it as well.  Since all of the \abclz models
  being mixed violate Theorem~\ref{thm:main-thm}, so must \abclo.
\end{proof}

\subsubsection{A note on transformations}

As Leifer notes, transformations also play a role in restricting the
structure of $\psi$-epistemic ontological models~\cite[Section
8.1]{Leiferquantumstatereal2014}. If an ontological model successfully
represents all unitary transformations, then
$\left|\braket{\psi}{\phi}\right|=\left|\braket{\psi'}{\phi'}\right|$
implies that $\ket\psi,\ket\phi$ are ontologically distinct if and
only if $\ket{\psi'},\ket{\phi'}$ are ontologically distinct. If the
model also includes all CPTP maps, then
$\left|\braket{\psi}{\phi}\right|\geq\left|\braket{\psi'}{\phi'}\right|$
implies that $\ket{\psi'},\ket{\phi'}$ are ontologically distinct if
$\ket\psi,\ket\phi$ are.

It immediately follows that the LJBR and \abclz models cannot
faithfully represent unitary transformations. In each model there
exist quantum states which are ontologically distinct from every other
state; pick one of these states, and it is easy to find examples of
pairs of ontologically distinct states with any inner product.

However, transformations cannot necessarily rule out the \abclo model:
since it is pairwise $\psi$-epistemic, the unitary condition could
\emph{in principle} be satisfied. That said, the transformation rule
would be complicated because it would have to map \emph{between}
models that are mixed together, so it is certainly an open question
whether this is actually possible.

\subsection{Models of subtheories}
Although we have dealt so far with models that include the full
quantum set of preparations, transformations, and measurements, there
is the possibility that we can retain $\psi$-epistemic models of
subtheories. It turns out that although the stabilizer subtheory can
be represented by a $\psi$-epistemic model, the more general Kitchen
Sink model which models any finite subtheory cannot in general
represent state update under measurement.

\subsubsection{Kitchen Sink model}
\label{sec:kitchen-sink-model}
The Kitchen Sink model is a $\psi$-epistemic ontological model for any
finite subtheory of quantum
theory~\mbox{\cite[Section~IIIC]{HarriganRepresentingprobabilisticdata2007}}. Given
a finite set of projective measurements $\M=\{M^{(i)}\}$, we choose
our ontic states to be a list of measurement outcomes
$\lambda=(\lambda_1,\lambda_2,\ldots,\lambda_\m)$. That is,
$\lambda_i=k$ means that if $M^{(i)}=\{\Pi^{(i)}_j\}$ is measured on
the ontic state $\lambda$, the outcome $\Pi_k$ will occur with
certainty.  For a system of dimension $d$, the maximum number of
projectors in any given measurement is $d$, so we pad all of our
measurements with $0$s until they have $d$ elements. The Kitchen Sink
model is then defined by
\begin{align}
  \Lambda            & = \Z_d^\m\nonumber\\
  \mu(\lambda|\psi)  & =\prod_{i=1}^\m\tr(\Pi^{(i)}_{\lambda_i}\proj\psi) \nonumber\\
  \xi(k|\lambda,M^{(i)}) & = \delta(k,\lambda_i)                 
\end{align}
The Kitchen Sink is pairwise $\psi$-epistemic for any subtheory, and
also never $\psi$-ontic if we include all pure states in our
subtheory. Transformations can additionally be modeled under the
assumption of a closed subtheory. 

We can only rule out the Kitchen Sink model for certain subtheories,
as it is easy to construct subtheories with trivial update rules
(e.g. by only including rank-1 measurements). That said, our
requirements are few and are satisfied by the multi-qu$p$it stabilizer
subtheory, arguably the most important subtheory of quantum theory.
Specifically, we only need to include two states
$\ket\alpha,\ket\beta$ and two measurements $M^{(1)}=\{\Pi,\Id-\Pi\},M^{(2)}$
satisfying
\begin{align}
  \braket{\alpha}{\beta} &\neq 0\label{eq:4}\\ 
  \ev{\Pi}{\alpha} &\neq 0\label{eq:5}\\
  \ev{\Pi}{\beta} &\neq 0\label{eq:8}\\
  M^{(2)}\text{ distinguishes }&\Pi\ket\alpha\text{ and }\Pi\ket\beta\label{eq:10}
\end{align}
The first is required because we don't expect orthogonal states to be
ontologically indistinguishable. The next two stipulate that there is
a nonzero chance of obtaining an outcome $\Pi$ when measuring
$\ket\alpha$ and $\ket\beta$, so that its support overlaps with their
support. The last condition implies that the post-measurement states
are ontologically distinct~\cite{Leiferquantumstatereal2014}.

\begin{theorem}
  For any finite subtheory containing states and measurements
  satisfying the conditions given in Eqs.~\ref{eq:4}--\ref{eq:10}, the
  Kitchen sink model cannot model state update under measurement.
\end{theorem}
\begin{proof}
  Theorem~\ref{thm:main-thm} implies that, if $\Pi$ maps
  $\ket\alpha,\ket\beta$ to ontologically distinct states, then
  \begin{equation}
    \label{eq:21}
    \int_\Lambda\dd{\lambda} \xi(\Pi|\lambda)\mu(\lambda|\alpha)\mu(\lambda|\beta) = 0.
  \end{equation}
  We evaluate this quantity for the Kitchen Sink's response functions
  and preparation distributions, using the states and measurement $M^{(1)}$ satisfying
  Eqs.~\ref{eq:4}--\ref{eq:10}:
  \begin{align}
    \int_\Lambda&\dd{\lambda} \xi(k=0|\lambda,M^{(1)})\mu(\lambda|\alpha)\mu(\lambda|\beta)      \nonumber\\
                & =\sum_{\lambda\in\Z_r^\m} \delta(0,\lambda_1)\prod_{j=1}^\m\tr(M^{(j)}_{\lambda_j}\proj\alpha)\tr(M^{(j)}_{\lambda_j}\proj\beta) \nonumber\\
                & = \sum_{\lambda_1\in\Z_r}\delta(0,\lambda_1)\tr(M^{(1)}_{\lambda_1}\proj\alpha)\tr(M^{(1)}_{\lambda_1}\proj\beta)\nonumber\\
                &\qquad\qquad\cdot\prod_{j=2}^\m\sum_{l\in\Z_r}\tr(M^{(j)}_l\proj\alpha)\tr(M^{(j)}_l\proj\beta)     \nonumber\\
                & = \tr(M^{(1)}_0\proj\alpha)\tr(M^{(1)}_0\proj\beta)\nonumber\\
                &\qquad\qquad\cdot\prod_{j= 2}^\m\sum_{l\in\Z_r}\tr(M^{(j)}_l\proj\alpha)\tr(M^{(j)}_l\proj\beta)   \label{eq:15}
  \end{align}
  This final expression will be zero if and only if at least one of
  its factors is 0. The first two factors are nonzero by
  Eqs.~\ref{eq:5} and~\ref{eq:8}. The rest of the factors are nonzero
  due to Eq.~\ref{eq:4} and the completeness condition on the
  measurements.  Thus Eq.~\ref{eq:15} is nonzero and we have a
  contradiction, so state update cannot be represented faithfully.
\end{proof}

This demonstrates that Theorem~\ref{thm:main-thm} can create trouble
even in subtheories. In particular, the stabilizer subtheory satisfies
the requirements in Eqs.~\ref{eq:4}--\ref{eq:10}, so it cannot be
modeled by the Kitchen Sink. That said, we can show that the
stabilizer subtheory still supports a $\psi$-epistemic interpretation
using other models.

\subsubsection{Qu$p$it stabilizer subtheory}
\label{sec:quopit-wigner-model}


We begin with the straightforward case of $n$ $p$-dimensional systems,
for $p$ an odd prime. In this case, the stabilizer subtheory has an
ontological model given by the discrete Wigner
function~\cite{VeitchResourceTheoryStabilizer2014,GrossHudsonTheoremfinitedimensional2006}
(see Appendix~\ref{sec:wigner-intro} for definitions of the stabilizer
subtheory and the phase-point operaters $A_\lambda$):
\begin{align}
  \Lambda &= \mathbb{Z}_p^n\times\mathbb{Z}_p^n\nonumber\\
  \mu(\lambda|P_\psi) &= \frac{1}{p^n}\tr(A_\lambda\proj\psi)\nonumber\\
  \xi(k|\lambda,M_{\{\Pi_i\}}) &= \tr(\Pi_kA_\lambda)
\end{align}
As the Wigner function is a quasi-probability
distribution~\cite{FerrieQuasiprobabilityrepresentationsquantum2011},
it necessarily takes on negative values if we try to model the full
quantum theory. If, however, we restrict to modeling preparations,
transformations, and measurements in the qu$p$it stabilizer subtheory,
then the representation is positive and it forms a well-defined
ontological
model~\cite{GrossNonnegativeWignerfunctions2007,GrossHudsonTheoremfinitedimensional2006}. It
is both pairwise $\psi$-epistemic and never $\psi$-ontic.

The stabilizer subtheory presents a challenge in that measuring a
single qu$p$it is described by a rank-$p^{n-1}$ measurement, which may
run into trouble due to Theorem~\ref{thm:main-thm}. In particular,
there are many examples in the stabilizer subtheory of the type of
measurements that broke the Kitchen Sink model. Nonetheless, we can
specify the update rule
\begin{align}
  \eta(\lambda'|k,\lambda,M_{\{\Pi_i\}}) &= \frac{1}{p^n}\frac{\tr(A_\lambda\Pi_k A_{\lambda'}\Pi_k)}{\tr(\Pi_kA_\lambda)}.
\end{align}
Note the normalization factor in $\eta$ which makes clear that $\eta$
is only defined in the support of $\xi$. The fact that this
successfully reproduces quantum statistics follows from the fact that
post-selected measurement is a completely positive map, and this is
how completely positive maps are represented in quasi-probability
representations~\cite{FerrieQuasiprobabilityrepresentationsquantum2011}.
$\eta$ is always positive for the stabilizer subtheory, though we do
not include the proof here.
 
For the case $p=2$, Lillystone and Emerson construct a 
$\psi$-epistemic model of the $n$-qubit stabilizer formalism that
successfully represents state update under
measurement~\cite{LillystoneContextualpsEpistemicModel2019}. This model starts
from the Kitchen Sink model and augments $\Lambda$ so that the
problematic overlaps of the Kitchen sink are removed. This model is
not pairwise $\psi$-epistemic, but a modified version (see appendix
of~\cite{LillystoneContextualpsEpistemicModel2019}) is never-$\psi$-ontic. We
don't present the construction here because it is significantly more
convoluted than the model above for $p\geq3$. This reflects the
often-observed ill-behaved nature of the qubit stabilizer subtheory.

Tangentially, if we extend the Wigner function to the full quantum
theory, we get negatively represented state update, as expected. One
consequence of this is that some state updates can't be normalized, so
the Wigner function state update must include a renormalization step
not allowed in ontological models or quasiprobability representations.
Although further discussion of state update under measurement in
quasi-probability representations is beyond the scope of this paper,
we note that Theorem~\ref{thm:main-thm} does not hold for
quasi-probability representations so this could be one potential
direction for related future work.

\section{Discussion}
\label{sec:discussion}
  
We have demonstrated that state update under measurement poses a
serious challenge to $\psi$-epistemic interpretations of quantum
theory in the ontological models framework: all currently known
$\psi$-epistemic models for full quantum theory in $d\geq3$ cannot
faithfully represent state update. This runs in direct contrast to the
prevailing view that $\psi$-epistemic models provide a compelling
explanation of state update.

There are a number of remaining open questions. Most pressingly, we
have re-opened the possibility of proving a general $\psi$-onticity
result without additional assumptions---will the methods of this paper be
useful in doing so?

On the one hand, the proofs above do not rule out the possibility of
extending the ontology of the broken models in order to represent
state update under measurement while still retaining the epistemicity
of the model. This is exactly the route taken
in~\cite{LillystoneContextualpsEpistemicModel2019} for the $n$-qubit
stabilizer subtheory. Granted, the $n$-qubit stabilizer subtheory was
brought within an inch of $\psi$-onticity by this process, so it seems
unlikely that a similar technique will work for the full quantum
theory.

In the other direction, we've shown that consideration of state update
puts powerful constraints on the structure of $\psi$-epistemic models.
These restrictions would ideally lead to a categorical statement like
``$\psi$-epistemic models cannot represent state-update,'' but there
are challenges to achieving this conclusion. In particular, we note
that all of the $\psi$-epistemic models that we considered share the
property of outcome determinism, which means that
Theorem~\ref{thm:main-thm} may be less trouble in
non-outcome-deterministic models. At the very least, any no-go theorem
will have to include measurements, states, and/or transformations from
outside the stabilizer subtheory since we have shown that
$\psi$-epistemic models for this subtheory exist.

What import does our result have for the general interpretational
project of quantum theory? First of all, we have demonstrated that
$\psi$-epistemicists have yet another challenge to overcome: a
successful explanation of state-update. This is in contrast with the
usual claim that this arena is one where epistemic interpretations
have an advantage over ontic interpretations. As we emphasized in the
introduction, our results only strictly apply to interpretations that
can be described by the ontological models formalism, but there may be
a qualitative message for epistemic and doxastic interpretations that
are outside this formalism as well.

\section{Acknowledgements}
This research was supported by the Canada First Research Excellence
Fund and the NSERC Discovery program. J.B.R. acknowledges funding from
an Ontario Graduate Scholarship. The authors would like to thank Matt
Leifer and David Schmid for productive discussions. We would also like
to thank two anonymous reviewers and our editor at Quantum, whose
remarks motivated the discussion in
Appendix~\ref{sec:additional-assumption}.

\bibliographystyle{unsrtnat} 
\bibliography{psi-onticity,contextuality,HMMs,various-foundations,just-for-onticity-paper}

\clearpage

\appendix

\section{Is state update an additional assumption?}
\label{sec:additional-assumption}

Throughout this paper, we have used the phrase `no additional
assumptions' in order to contrast our results with those of existing
no-go theorems. By this we don't mean to raise an argument over the
semantics of exactly what one means by `assumption,' but rather we
mean to say that the consideration of state update, which we propose
may be able to categorically rule out $\psi$-epistemic models, is very
different from the `additional assumptions' used in existing no-go
theorems.

To make this precise, we can break up our consideration of state
update into two parts. The first part is simply the fact that we are
considering a broader set of empirical situations than considered
previously: experiments involving sequential measurements. The
ontological models formalism, as originally defined, does not have the
tools to model this empirical situation, so we are forced to extend
the formalism to do so. This leads to the second part, which is an
assumption about \emph{how} we ought to model this empirical
situation. While, in principle, there are many ways one can do this,
we show in Appendix~\ref{sec:hmm-appendix} that there is only one way
to do this that is consistent with the broader conceptual
underpinnings (i.e. pre-existing assumptions) of the ontological
models formalism. By examining the core assumptions of the ontological
models formalism rather than its strict mathematical formulation, we
show that all of the structure associated with state update follows
directly from the pre-existing assumptions of the ontological models
formalism.

This is in contrast to the assumptions used in previously published
no-go theorems. First of all, they start from empirical situations
that can be described by the vanilla ontological models
formalism. Then, they impose a particular restriction on the form of
the ontological models that happens directly at the level of the ontic
state; for example, the PBR
theorem~\cite{Puseyrealityquantumstate2012} assumes the preparation
independence postulate (PIP) which says that says that, for two
subsystems $A$ and $B$, if the quantum state is a product state, then
the preparation distribution $\mu$ is a product distribution:
\begin{align}
  \label{eq:2} \Lambda_{AB} &= \Lambda_A\times\Lambda_B\\
&\text{and}\nonumber\\ \Psi_{AB}=\Psi_A\otimes\Psi_B
\quad&\nonumber\\\implies
\mu(\lambda_A,\lambda_B|\Psi_{AB})&=\mu(\lambda_A|\Psi_A)\cdot\mu(\lambda_B|\Psi_B).
\end{align} We can see from this symbolic expression that this is
indeed a direct assumption about the structure of the ontological
model, rather than an empirical consideration. More to the point, it
imposes structure on the ontological model that is not required by the
bare formalism.

As another example, the Colbeck-Renner
argument~\cite{Colbecksystemwavefunction2017} claims to not make any
assumptions of the kind described in the PIP; they do, however, assume
a spacetime structure which is, again, additional to the ontological
models formalism. Spacetime and locality arguments are also what
motivate the PIP. These relativistically-motivated considerations are
interesting in their own right, but since Bell's theorem already rules
out $\psi$-epistemic models based on reasonable spacetime structure
(i.e. the locality assumption), the goal of our paper is to arrive at
a similar conclusion \emph{without} imposing spacetime structure. In
particular, a $\psi$-epistemicity no-go theorem would directly imply
the results of Bell's theorem and is thus a stronger
result~\cite{Leiferquantumstatereal2014}.

Whether \emph{operational} state update, i.e. the L\"uders rule, is
considered an assumption of some kind or a consequence of other
operational axioms is irrelevant to our analysis.  The key point for
our analysis is that the state-update rule is an experimentally
validated operational feature/prediction of QM theory, much like
entanglement or unitary dynamics, and therefore is an empirical
feature of the QM framework that the framework of ontological models
ought to explain.  So just like the discovery from Bell that entangled
states lead to new constraints/insights into the necessary features of
($\psi$-epistemic) ontological models, similarly our insight is that
the dynamics associated with state-update also leads to new
constraints/insights into the necessary features of ($\psi$-epistemic)
ontological models.

\section{Justifying the form of the state update rule: ontological
  models as hidden Markov models of stochastic channels}
\label{sec:hmm-appendix}

In the context of state update under measurement, it is illuminating
to motivate the definition of an ontological model from the point of
view of the hidden Markov models (HMM) literature. We do this in order
to (a) provide a rigorous treatment of multiple-measurement scenarios
presented informally in Section~\ref{sec:ont-models} and (b) clarify
the assumptions that define the ontological models framework. In the
process, we use these assumptions to show that the form of our state
update rule is the only one consistent with the already-defined
components of the ontological models formalism.

We picture a quantum circuit as a memoryful stochastic channel
(Fig.~\ref{fig:circuit-channel}). The channel that we often discuss
with regards to a quantum circuit is the (quantum) channel that takes
the input quantum state and maps it to the output quantum state. For
present purposes, we will instead think of it as a channel from the
experimenter to individual measurement outcomes used repeatedly at
each time step. Pictorially, one might think of this as `rotating the
channel ninety degrees' in a circuit diagram.

\begin{figure}
  \centering
  \begin{tabular}{c}
    \Qcircuit @C=.5cm @R=.2em {
    &                     &                  &                   & \ustick{k} \cwx[1]   &                    & \ustick{k'} \cwx[3]   &     \\
    & \multimeasure{5}{P} & \multigate{3}{U} & \qw               & \multimeasureD{1}{M} & \qw                & \qw                   & \qw \\
    & \ghost{P}           & \ghost{U}        & \qw               & \ghost{M}            & \multigate{3}{U''} & \qw                   & \qw \\
    & \ghost{P}           & \ghost{U}        & \qw               & \qw                  & \ghost{U''}        & \multimeasureD{3}{M'} & \qw \\
    & \ghost{P}           & \ghost{U}        & \multigate{2}{U'} & \qw                  & \ghost{U''}        & \ghost{M'}            & \qw \\
    & \ghost{P}           & \qw              & \ghost{U'}        & \qw                  & \ghost{U''}        & \ghost{M'}            & \qw \\
    & \ghost{P}           & \qw              & \ghost{U'}        & \qw                  & \qw                & \ghost{M'}            & \qw
                                                                                                                                       }
  \end{tabular}

  \vspace{.5cm}
  \begin{tikzpicture}[node distance = 1.3cm, auto]
    \node (k0) {$k_0=0$};
    \node [below of=k0,node distance = 1.7cm] (a0) {$a_0=P$};

    \node [right of=k0] (k1) {$k_1=0$};
    \node [right of=a0] (a1) {$a_1=U$}; 

    \node [right of=k1] (k2) {$k_2=0$};
    \node [right of=a1] (a2) {$a_2=U'$}; 

    \node [right of=k2] (k3) {$k_3=k$};
    \node [right of=a2] (a3) {$a_3=M$};
    
    \node [right of=k3] (k4) {$k_4=0$};
    \node [right of=a3] (a4) {$a_4=U''$}; 

    \node [right of=k4] (k5) {$k_5=k'$};
    \node [right of=a4] (a5) {$a_5=M'$};

    \node[above left of=k0, node distance = 1cm] (line0) {};
    \node[above right of=k5,node distance = 1cm] (line1) {};
    \draw [dashed] (line0) -- (line1);
    \node [single arrow,rotate=90, draw=none, font=\scriptsize,text width = .8cm, fill=violet!50] at (3.25,-.95) {channel};
  \end{tikzpicture}
  \caption{A quantum circuit can be pictured as a stochastic channel,
    as described in the text. The inputs to the channel $a_t$ are the
    choice of operation, and the outputs $k_t$ report the results of
    measurements.}
  \label{fig:circuit-channel}
\end{figure}
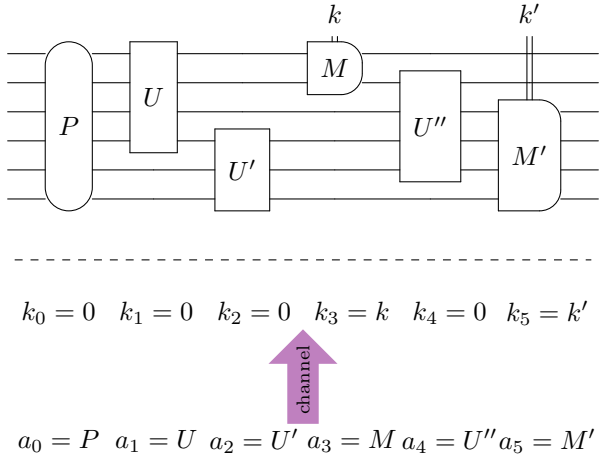

The input string $\olra{a_0}=\ldots a_{-2}a_{-1}a_0a_1a_2\ldots$ of
the channel is the experimenter's choice of action, which we take to
be a sequence of operational preparation, transformation, or
measurement procedures.  The output string $\olra{k_0}$ of the channel
reports either the results of measurements or a trivial output for
preparations and transformations. The subscript $0$ indicates in both
cases the time that we take as an origin/reference point. The channel
is then described by the conditional probability distribution
$\Pr(\olra{k_0}|\olra{a_0})$. Following~\cite{BarnettComputationalMechanicsInputOutput2015},
we denote substrings with $a_{t:t+L}=a_ta_{t+1}\ldots a_{t+L-1}$, and
also define the past $\ola{a_t}=a_{-\infty:t}$ and future
$\ora{a_t}=a_{t:\infty}$. We can now define two properties of
stochastic channels:
\begin{definition}[Stationary]
  \label{def:stationary}
  A \emph{stationary} channel is one that has time-translation
  symmetry, so statistics are not affected by our choice of
  time-origin:
  \begin{align}
    \label{eq:34}
    \Pr(k_{t:t+L}|\olra{a_t})=\Pr(k_{0:L}|\olra{a_o}) \qquad \text{and} \nonumber\\
    \Pr(\olra{k_t}|\olra{a_t})= \Pr(\olra{k_0}|\olra{a_0}) \qquad\forall t,L,\olra{a}.
  \end{align}
\end{definition}
\begin{definition}[Causal]
  \label{def:causal}
  A \emph{causal} channel is one for which a finite output substring
  depends only on input symbols in its past:
  \begin{equation}
    \Pr(k_{t:t+L}|\olra{a}) = \Pr(k_{t:t+L}|\ola a_{t+L}).\qquad\forall t,L,\ola{a}\label{eq:35}
  \end{equation}
\end{definition}
It is shown in~\cite{BarnettComputationalMechanicsInputOutput2015}
that a channel satisfying these two properties can be specified
entirely by the single-symbol recurrence relation
\begin{equation}
\label{eq:29}
  \Pr(k_0|a_0,\ola{a_0}, \ola{k_0}).
\end{equation}
Note that this does not imply a Markov process, since it depends in
general on the entire histories $\ola{a_0}$ and $\ola{k_0}$. All we
mean by single-symbol is that we are not specifying the probabilities
over the whole future, just a single output symbol. We now construct
an HMM as follows:

\begin{definition}[Hidden Markov Model]
  \label{def:hmm}
  A \emph{hidden Markov model} (HMM) of a stationary, causal channel
  is specified by an additional random variable $\lambda$ taking
  values in a state space $\Lambda$. It is given a joint probability
  distribution over $\olra{\lambda_0},\olra{a_0},\olra{k_0}$ so that
  the recurrence relation above (Eq.~\ref{eq:29}) becomes
\begin{equation}
  \label{eq:30}
  \Pr(k_0,\lambda_1|a_0,\lambda_0,\ola a_0, \ola k_0,\ola \lambda_0)
  =\Pr(k_0,\lambda_1|a_0,\lambda_0).
\end{equation}
\end{definition}
In other words, the state $\lambda$ renders the future conditionally
independent of the past and induces a Markov process over the state
space $\Lambda$ that mediates the channel statistics. This is the
property of the ontological models formalism called
``$\lambda$-mediation''
in~\cite{Leifertimesymmetricinterpretation2017}. An influence
diagram~\cite{HowardInfluenceDiagrams2005} of a stationary, causal
channel is shown in Figure~\ref{fig:influence-diagram} before and
after the specification of an HMM.

To see that specification of an HMM as in Eq.~\ref{eq:30} is
equivalent to the definition of an ontological model given in
Section~\ref{sec:ont-models}, we first note that generally we don't
think of preparations and transformations having output; to account
for this, we stipulate that they give a trivial, deterministic output
$k_0=0$. We then factor the probability distribution from
Eq.~\ref{eq:30} and look separately at the cases where $a_0$ is a
preparation, transformation, or measurement:
\begin{align} 
  \Pr&(k_0,\lambda_1|\lambda_0,a_0)\nonumber\\
  &= \Pr(\lambda_1|k_0,\lambda_0,a_0)\Pr(k_0|\lambda_0,a_0)\nonumber\\
  &=
    \begin{cases}
      \mu(\lambda_1|P) \delta_{k_0,0} & a_0 =P\in \mathcal P\\
      \Gamma(\lambda_1|\lambda_0,T) \delta_{k_0,0} & a_0 =T\in \mathcal T\\
      \eta(\lambda_1|k_0,\lambda_0,M)\xi(k_0|\lambda_0,M) & a_0 =M\in \mathcal M
    \end{cases}
                                                            \label{eq:38}
\end{align}
where $\delta$ is the Kronecker delta. This can be seen as the
rigorous definition of $\eta$ which emerges naturally from this
recognition that ontological models are equivalent to hidden Markov
models. In particular, this gives a mathematical reason why $\eta$ is
only defined in the support of $\xi$. If we take the joint
distribution $\Pr(k_0,\lambda_1|\lambda_0,a_0)$ to be the more
fundamental object, then it is clear we can obtain $\xi$ directly by
marginalization
\begin{equation}
  \label{eq:37}
  \xi(k_0|\lambda_0,M)=\int_\Lambda \dd{\lambda_1} \Pr(k_0,\lambda_1|\lambda_0,M)
\end{equation}
which is always well defined, and then find $\eta$ by rearranging
Eq.~\ref{eq:38}:
\begin{equation}
  \label{eq:39}
  \eta(\lambda_1|k_0,\lambda_0,M) = \frac{\Pr(k_0,\lambda_1|\lambda_0,M)}{\xi(k_0|\lambda_0,M)}
\end{equation}
Thus clearly $\eta(\lambda_1|k_0,\lambda_0,M)$ is only well-defined
when $\xi(k_0|\lambda_0,M)\neq0$, which is how we defined its support.

\def\edgelength{1.8cm}
\def\diaglength{1.342cm}
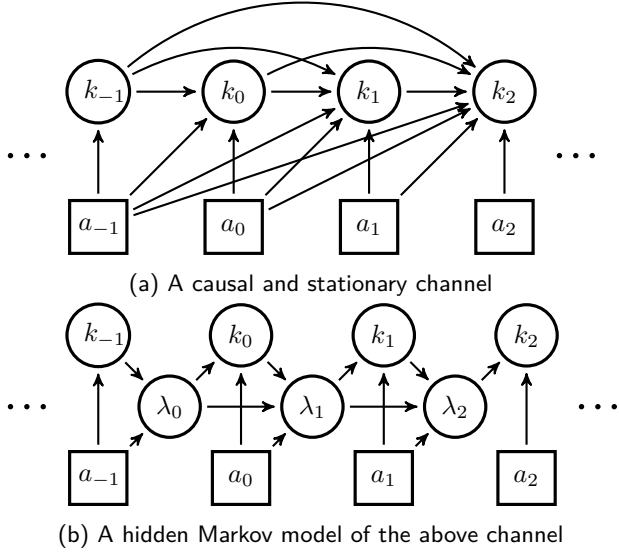
\begin{figure}
  \centering
  \begin{subfigure}[t]{0.48\textwidth}
    \begin{tikzpicture}[node distance = \edgelength, auto]
      \node [chance] (k2) {$k_2$};
      \node [chance, left of=k2] (k1) {$k_1$}
      edge[pil] (k2);
      \node [chance, left of=k1] (k0) {$k_0$}
      edge[pil,bend left=27] (k2)
      edge[pil] (k1);
      \node [chance, left of=k0] (k-1) {$k_{-1}$}
      edge[pil,bend left=40] (k2)
      edge[pil,bend left=27] (k1)
      edge[pil] (k0);
      \node [decision, below of=k2] (a2) {$a_2$}
      edge[pil] (k2);
      \node [decision, below of=k1] (a1) {$a_1$}
      edge[pil] (k2)
      edge[pil] (k1);
      \node [decision, below of=k0] (a0) {$a_0$}
      edge[pil] (k2)
      edge[pil] (k1)
      edge[pil] (k0);
      \node [decision, below of=k-1] (a-1) {$a_{-1}$}
      edge[pil] (k2)
      edge[pil] (k1)
      edge[pil] (k0)
      edge[pil] (k-1);
      \node [above left of=a-1,node distance=\diaglength] {\Huge...};
      \node [above right of=a2,node distance=\diaglength] {\Huge...};
    \end{tikzpicture}
    \caption{A causal and stationary channel}
    \label{fig:stat-causal-sp}
  \end{subfigure}
  \begin{subfigure}[t]{0.48\textwidth}
    \begin{tikzpicture}[node distance=\diaglength, auto]
      \node [chance] (k2) {$k_2$};
      \node [chance, below left of=k2] (l2) {$\lambda_2$}
      edge[pil] (k2);
      \node [decision, below right of=l2] (a2) {$a_2$}
      edge[pil] (k2);

      \node [chance, above left of=l2] (k1) {$k_1$}
      edge[pil] (l2);
      \node [chance, below left of=k1] (l1) {$\lambda_1$}
      edge[pil] (k1)
      edge[pil] (l2);
      \node [decision, below right of=l1] (a1) {$a_1$}
      edge[pil] (k1)
      edge[pil] (l2);

      \node [chance, above left of=l1] (k0) {$k_0$}
      edge[pil] (l1);
      \node [chance, below left of=k0] (l0) {$\lambda_0$}
      edge[pil] (k0)
      edge[pil] (l1);
      \node [decision, below right of=l0] (a0) {$a_0$}
      edge[pil] (k0)
      edge[pil] (l1);

      \node [chance, above left of=l0] (k-1) {$k_{-1}$}
      edge[pil] (l0);
      \node [decision, below left of=l0] (a-1) {$a_{-1}$}
      edge[pil] (k-1)
      edge[pil] (l0);

      \node [above left of=a-1] {\Huge...};
      \node [above right of=a2] {\Huge...};

    \end{tikzpicture}
    \caption{A hidden Markov model of the above channel}
    \label{fig:hmm} 
  \end{subfigure}

  \caption{Influence diagrams~\cite{HowardInfluenceDiagrams2005} for
    stochastic channels. The boxes represent choices made by the
    experimenter, circles represent random variables, and arrows
    represent a possible causal influence. Note that no arrows point
    backwards in time, and that in the hidden Markov model the state
    $\lambda_t$ mediates all causal influences through time.}
  \label{fig:influence-diagram}
\end{figure}

Definitions~\ref{def:stationary}--\ref{def:hmm} constitute an
equivalent formulation of the ontological models formalism. The
assumptions of this construction can be broken down as follows: (a)
quantum theory is described by a stochastic channel, (b) this channel
is stationary, (c) it is causal, and (d) we assign the system a state
which acts as an HMM of the channel. The authors
of~\cite{Leifertimesymmetricinterpretation2017} identify (c) and (d),
calling them non-retrocausality and $\lambda$-mediation,
respectively. Assumptions (a) and (b) were implicitly present but not
explicitly identified. Since our formulation of the state update rule
follows directly from these assumptions, we see that it is not
`additional' to the ontological models formalism, but a unique
extension to describe a more general empirical scenario.

\section{A brief introduction to the stabilizer subtheory}
\label{sec:wigner-intro}
We focus here on the stabilizer subtheory for $n$ qu$p$its, where $p$
is a prime. For a more detailed exposition, we refer the reader
to~\cite{GrossHudsonTheoremfinitedimensional2006,VeitchResourceTheoryStabilizer2014}.

\paragraph{Mathematical objects} The generalizations of the $X$ and
$Z$ operators to a single qu$p$it are defined by their action on the
standard basis $\{\ket j\}$ for $j=0,1,\ldots,p-1$:
\begin{align}
  X\ket j &= \ket{j+1}\\
  Z\ket j &= \omega^j\ket j\\
  \omega &= e^{2\pi i/p}
\end{align}
All integer arithmetic is done $\bmod p$. Then the full set of
generalized Pauli operators on $n$ qu$p$its is given by
\begin{align}
  T_{(\vb x,\vb z)}
  &=
    \begin{cases}
      \bigotimes_{j=0}^{n-1} X^{x_j}Z^{z_j} & p = 2\\
      \bigotimes_{j=0}^{n-1} \omega^{x_j z_j/2}X^{x_j}Z^{z_j} & p > 2
    \end{cases}\\
  \text{for } \vb x &= (x_0,x_1,\ldots,x_{n-1})\in\Z_p^n \nonumber\\
  \text{and } \vb z &= (z_0,z_1,\ldots,z_{n-1})\in\Z_p^n\nonumber
\end{align}
We also define the symplectic inner product as
\begin{equation}
  \label{eq:26}
  \left[ (\vb x, \vb z), (\vb x',\vb z')\right]
  = \vb z\vdot \vb x' - \vb x\vdot \vb z'.
\end{equation}
Finally, the phase-point operators $A_\lambda$, for
$\lambda=(\vb x,\vb z)\in \Z_p^n\times\Z_p^n$, are a symplectic
Fourier transform of the Pauli operators:
\begin{equation}
  \label{eq:27}
  A_\lambda = \frac{1}{p^n}\sum_{\lambda'\in\Z_p^n\times\Z_p^n}\omega^{[\lambda,\lambda']}T_{\lambda'}
\end{equation}

\paragraph{The stabilizer subtheory} A stabilizer group $S$ is a set
of $p^n$ mutually commuting Pauli operators, which can be specified by
a set of $n$ generators. There is a unique state (up to global phase)
which is an eigenvector of all of these operators with eigenvalue
$+1$; we say that $S$ \emph{stabilizes} this state. For example, for
two qubits, the Bell state
\begin{equation}
  \label{eq:24}
    \ket\Psi = \ket{00}+\ket{11}
\end{equation}
is stabilized by
\begin{align}
  \label{eq:25}
  S_\Psi &= \langle Z_1Z_2,X_1X_2 \rangle\\
         &=\{\Id,Z_1Z_2,X_1X_2,-Y_1Y_2\}.
\end{align}
Here a subscript indicates on which qubit the operator is acting,
e.g. $Z_1=Z\otimes\Id$ describes $Z$ acting on the first qubit. A
stabilizer state, then, is a state which is stabilized by a group of
$p^n$ Pauli operators.
 
Stabilizer measurements are simply measurements of the Pauli
operators. Note that since each Pauli operator has $p$ eigenvalues,
these amount to a measurement of $p$ projectors, each with rank
$p^{n-1}$. Lower-rank projectors can be constructed by performing
commuting Pauli measurements sequentially.

Finally, the transformations of the stabilizer subtheory are called
Clifford transformations. These are the transformations that map the
set of Pauli operators to itself, up to a global phase. In other
words, it is the normalizer of the Pauli group.

The stabilizer subtheory thus consists of preparations corresponding
the set of stabilizer states, measurements of Pauli observables, and
Clifford transformations, along with convex combinations thereof.

\end{document}